\documentclass[journal,twoside,web]{ieeecolor}
\usepackage{generic}
\usepackage{cite}
\usepackage{amsmath,amssymb}

\usepackage{graphicx}
\usepackage[ruled,linesnumbered,vlined]{algorithm2e}

\usepackage{bm}
\usepackage{textcomp}
\usepackage{enumerate}
\usepackage{subfigure}

\usepackage[T1]{fontenc}
\def\BibTeX{{\rm B\kern-.05em{\sc i\kern-.025em b}\kern-.08em
    T\kern-.1667em\lower.7ex\hbox{E}\kern-.125emX}}
\graphicspath{{./fig}}

\newtheorem{theorem}{Theorem}
\newtheorem{lma}{Lemma}
\newtheorem{exam}{Example}
\newtheorem{corollary}{Corollary}

\newtheorem{proposition}{Proposition}

\newtheorem{assumption}{Assumption}
\newtheorem{remark}{Remark}

\DeclareMathOperator{\diag}{\mathrm{diag}}
\newcommand{\p}{p}
\newcommand{\ca}{c}
\newcommand{\en}{e}
\newcommand{\de}{d}
\newcommand{\anti}{n_{a}}
\newcommand{\ctotal}{C_{\mathrm{total}}}
\newcommand{\tf}[1]{\bm{#1}}

\newcommand{\hua}[1]{\mathsf{#1}}

\DeclareMathOperator{\TE}{\mathbf{H}}

\newenvironment{myproof}[1][]
{\if\relax\detokenize{#1}\relax
	\textit{Proof: }%
	\else
	\textit{Proof of #1: }%
	\fi}
{\hfill\scalebox{0.855}{$\blacksquare$}}

\begin{document}
\title{Linear Coding of LTI Sources Over Vector Gaussian Channels: A Majorization Approach}     
\author{Shihao Jin, Junhui Li, Shinji Hara, \IEEEmembership{Fellow, IEEE}, Wei Chen, \IEEEmembership{Member, IEEE}
\thanks{{S. Jin and J. Li contributed equally to this work.}}
\thanks{S. Jin and W. Chen are with the School of Advanced Manufacturing and Robotics
\& the State Key Laboratory for Turbulence and Complex Systems, Peking University, Beijing 100871, China  (e-mail: jinshihao@stu.pku.edu.cn, w.chen@pku.edu.cn).}
\thanks{J. Li is with the School of Electrical Engineering, Guangxi University, Guangxi 530000, China (e-mail: jh.li@gxu.edu.cn).}
\thanks{S. Hara is with the Supercomputing Research Center, Institute of Integrated Research,  Institute of Science Tokyo, Tokyo 152-8550, Japan (e-mail:
shinjihara5202@gmail.com).}
}

\maketitle

\begin{abstract}
We study the design of linear time-invariant (LTI) encoder-decoder pairs for transmitting the state of a discrete-time LTI vector source over power-constrained parallel Gaussian channels with feedback.
Two types of power constraints are considered. 
Under individual subchannel power constraints, a necessary and sufficient condition for designing an encoder-decoder pair that achieves bounded estimation error covariance (EEC) is established via two coupled majorization inequalities involving the subchannel signal-to-noise ratios and the antistable poles of the source.  
Under total channel power constraint, we derive the minimum total power required for a feasible encoder-decoder design by exploiting partial-order progamming under majorization order.
An analytical optimal power allocation is obtained for the case of equal noise variances,  which admits a water-filling interpretation; for general noise case, a sequential water-filling algorithm is developed.
Our results reveal that the
difficulty of transmitting a discrete-time LTI source via LTI coding is governed not only by its topological entropy, but also by the evenness of the log-magnitudes of its antistable poles.
 The design methods for feasible encoder-decoder pairs are also provided.
\end{abstract}


\section{Introduction}
\label{sec:introduction}
\IEEEPARstart{T}{ransmitting} the state of a Gauss-Markov source over noisy communication channels is a basic problem in networked estimation and control. Particularly, an unstable source continuously generates uncertainty through its antistable poles, so sufficient communication resources are needed to prevent the estimation error covariance from diverging. Characterizing the amount of communication resources required for remote estimation of an unstable source is thus a fundamental research problem. This problem is important because remote estimation is an indispensable component of output-feedback networked control. It is also relevant in practical scenarios where a control system may temporarily become open loop and unstable due to communication disruption or adversarial interference, so that monitoring the source state becomes vital for taking timely countermeasures.

A seminal result in the networked control area is the ``data-rate theorem''. A version of it states that the minimum communication resource required for networked stabilization of a single-input LTI system via state feedback or a single-input single-output (SISO) minimum-phase LTI system via output feedback is given by the topological entropy of the open-loop plant, namely, the sum of the real parts of unstable poles for the continuous-time systems, or the sum of the log-magnitudes of unstable poles for the discrete-time systems. The data-rate theorem was established under various information constraints \cite{Nair2003exponential,tatikonda:2004control,elia2001stabilization,fu2005sector,braslavsky:2007feedback} and stimulated extensive research efforts worldwide. See \cite{Hespanha2007NCSs,Franceschetti2023Facets} and the references therein.

Estimation over communication networks have been studied under a variety of channel models, such as Gaussian channels, packet-dropping channels, and data-rate constraints. See \cite{bansal1989simultaneous,li2011optimal,han2023coded,han2024coded,tanaka2024continuous,sinopoli2004kalman,Chen2023Optimal,liberzon2018entropy}, to name only a few. Pertinent to this paper are works on coding and estimation of LTI sources over power-constrained Gaussian channels. An early exploration in this direction is \cite{bansal1989simultaneous}, where a joint communication-control problem was transformed into an equivalent state estimation problem for a scalar source transmitted over a scalar Gaussian channel with feedback. It was shown that innovation-based linear transmission is optimal for first-order models and optimal within affine policies for higher-order models. In \cite{li2011optimal}, coding of vector LTI sources over SISO Gaussian channels with feedback was considered and a lower bound on the signal-to-noise ratio (SNR) required for the convergence of the estimation error was obtained.

Recent studies have further considered coding and estimation over multi-input multi-output (MIMO) Gaussian channels.
For the continuous-time LTI sources, \cite{jin2024remote} investigated remote estimation over parallel feedback Gaussian channels using an LTI encoder-decoder structure, revealing the interplay between the topological entropies of the plant's cyclic subsystems and the channel qualities in achieving bounded error covariance.
For the discrete-time setting, \cite{han2023coded} studied linear coding of an LTI source (with distinct unstable poles) over either a multi-input single-output or  single-input multi-output additive white Gaussian noise (AWGN) channel with noiseless feedback. It was shown that bounded estimation error covariance can be achieved if and only if the channel capacity is larger than the topological entropy of the source. 
The work \cite{han2024coded} then extended the setting to general MIMO AWGN channels with feedback and derived a sufficient condition for achieving bounded error covariance given in terms of a match between the subchannel capacities and a partition of the unstable poles assigned to the respective subchannels. The work \cite{han2024coded} was revisited in \cite{jahangiri2025finite}, where a solvability condition was derived through a nonconvex feasibility problem.

Related insights have also been developed in networked stabilization of MIMO systems over Gaussian channels, where remote estimation is embedded in the feedback control structure \cite{shu2011stabilization,VARGAS20133133,Zaidi2014Stabilization,zaidi2016tightness}. In \cite{shu2011stabilization}, a lower bound on the total transmission power for networked
stabilization of a discrete-time system was obtained and shown to be not always achievable by LTI encoder/decoder. In \cite{VARGAS20133133}, the authors
established necessary and sufficient conditions for networked stabilizability
over parallel SNR constrained AWGN channels in terms of unstable poles of the
plant along with their directions. In \cite{zaidi2016tightness}, stabilization over parallel Gaussian channels with a total power constraint was studied under noiseless feedback from the controller to both the sensor and the plant. The authors derived a sufficient stabilization condition in terms of a weak majorization relation between the plant's unstable poles and the subchannel capacities (although not stated explicitly in majorization form in the paper), through a linear time-varying coding strategy.

In this paper, we concentrate on LTI coding of discrete-time MIMO LTI sources over parallel Gaussian channels with feedback under both individual subchannel power constraints and total power constraints. Considering LTI coding is mainly due to its simple structure and easy implementation. Despite the efforts reviewed above, a complete and tractable solvability theory is still lacking. Under the individual subchannel power constraints, a precise characterization of how the subchannel qualities should be compatible with the unstable poles of the source remains to be discovered; under the total power constraint, the exact minimum required power and the corresponding optimal power allocation also remain to be characterized. These issues are addressed in this paper.

The main contributions of this paper are summarized below. 

First, under individual subchannel power constraints, we derive a necessary and sufficient solvability condition in terms of two coupled majorization inequalities. These majorizaion ineuqalities reveal that the distribution of SNRs or capacities of the subchannels and how it is matched with the distribution of log-magnitudes of the unstable poles of the source are critical in determining the solvability of the problem.

Second, under a total power constraint, we determine the exact minimum total power required for bounded estimation error covariance. For equal noise variances, the optimal allocation admits an analytic expression through the join operation in a majorization lattice and has a water-filling interpretation. The idea carries over to the general case of unequal noise variances, for which a sequential water-filling algorithm is developed.

Third, the results are constructive: when the corresponding solvability condition holds, we provide a systematic procedure for constructing an admissible LTI encoder-decoder pair.

The rest of this paper is organized as follows. 
Section~\ref{sec:problem_formulation} formulates the problem. Section~\ref{sec:preliminary} presents some preliminary knowledge. Section~\ref{sec:subchannel} and \ref{sec:total_power} give the main results. Section~\ref{sec:numerical_example} provides a numerical example. Some concluding remarks are given in Section~\ref{sec:conclusion}.

\section{Problem Formulation}\label{sec:problem_formulation}

We study the transmission of a vector source over a parallel Gaussian channel with feedback via an LTI encoder-decoder design. The considered configuration is shown in Fig.~\ref{fig:formulation}, and the detailed problem formulation is given below.

\begin{figure}[htbp]
	\centering
	\includegraphics[width=1\linewidth]{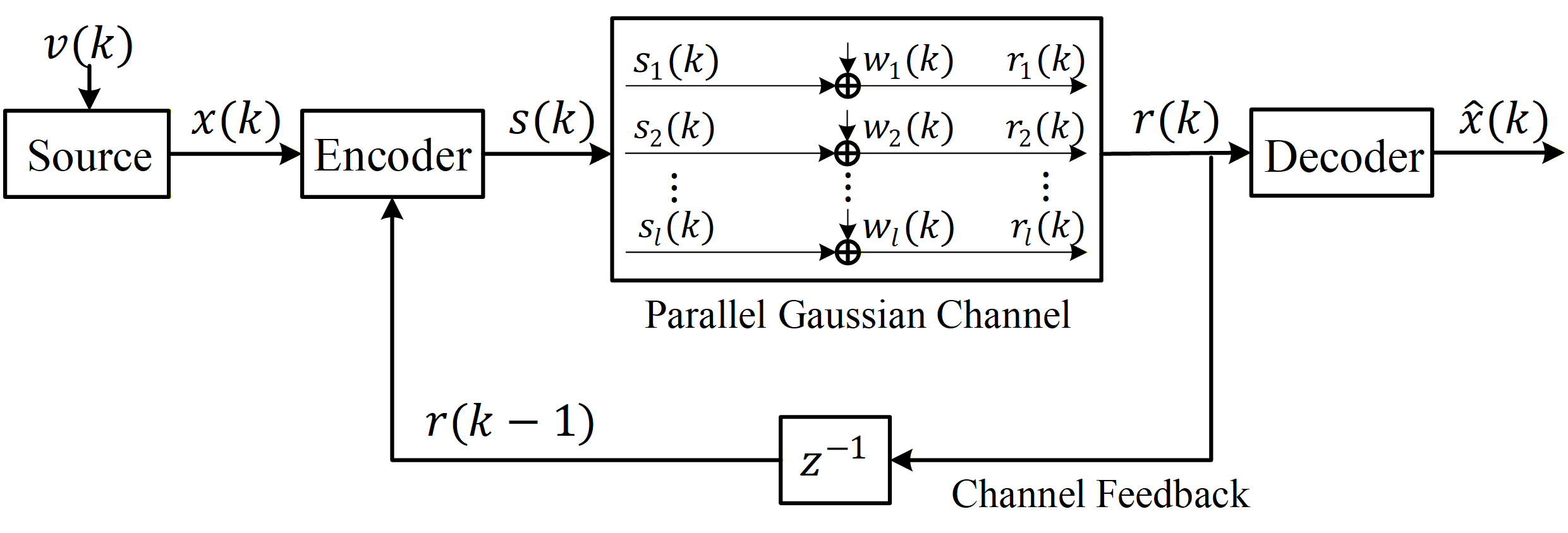}
	\caption{Linear coding of a discrete-time vector LTI source over a parallel Gaussian channel with feedback.}
	\label{fig:formulation}
\end{figure}

Consider a discrete-time LTI source
\begin{equation}\label{eq:source}
	x(k+1)=Ax(k)+v(k),
\end{equation}
where $A\in\mathbb{R}^{n\times n}$, $x(k)$ is the source state, and $v(k)$ is a zero-mean white vector noise with covariance matrix $V$. Assume that the source is unstable.

For the sake of simplicity, we make the following technical assumption throughout.

\begin{assumption}\label{asmpt:cyclic}
All unstable eigenvalues of $A$ have geometric multiplicity $1$.
\end{assumption}

Equivalently, Assumption \ref{asmpt:cyclic} requires that the direct sum of Jordan blocks associated with the unstable eigenvalues of $A$ be cyclic \cite{gantmakher2000theory}. Such assumption includes the case when $A$ has distinct eigenvalues as a special case.

The encoder processes the source state $x(k)$ and the one-step delayed signal $r(k-1)$ obtained via noiseless channel feedback to produce the encoded signal $s(k)$ for transmission. 
We adopt an LTI encoder:
\begin{equation}\label{eq:encoder}
	\!\!\begin{bmatrix}
		x_{\en}(k+1)\\
		s(k)
	\end{bmatrix}=\left[\begin{array}{ccc}
		A_{\en} & K & L\\
		M & N & R 
	\end{array}\right]\!\begin{bmatrix}
		x_{\en}(k)\\
		x(k)\\
		r(k-1)
	\end{bmatrix}\!,\, x_{\en}(0)=0,
\end{equation}
where $x_{\en}(k)$ is the state of the encoder.
We set $r(-1)=0$.

The encoded signal $s(k)$ is then transmitted to the decoder across a parallel AWGN channel consisting of $l$ subchannels. The channel output is given by 
\begin{equation*}\label{eq:channel}
	r(k)=s(k)+w(k),
\end{equation*}
where $w(k)\!=\!\begin{bmatrix}w_1(k)\!&\!\!\!\cdots\!\!\!&\!w_l(k)\end{bmatrix}'$ is a zero-mean  vector white Gaussian  noise with covariance matrix ${\Sigma}\!=\!\diag\!\{\sigma_1^2,\dots,\sigma_l^2\}$$>0$. The channel output is received by the decoder and is also returned to the encoder with one step delay. Assume that the channel noise $w$ is independent of the process noise $v$.

The decoder generates an estimate $\hat{x}(k)$ of the source state using the received signal $r(k)$.
We adopt an LTI decoder:
\begin{equation}\label{eq:decoder}
	\begin{bmatrix}
		x_{\de}(k+1)\\
		\hat{x}(k)
	\end{bmatrix}=\left[\begin{array}{ccc}
		A_{\de} & B\\
		C & D 
	\end{array}\right]\!\begin{bmatrix}
		x_{\de}(k)\\
		r(k)\\
		\end{bmatrix},\,x_{\de}(0)=0,
\end{equation}
where $x_{\de}(k)$ is the state of the decoder.

Define the estimation error of the source state as 
\[
\tilde{x}(k)=x(k)-\hat{x}(k).
\]
The estimation error covariance (EEC) is said to be bounded if 
for any initial source state, $\mathbb{E}\{\tilde{x}(k)\tilde{x}(k)'\}$ converges to some positive semi-definite matrix.

We are interested in the following problems.

\begin{itemize}
    \item
\textbf{Problem 1:} Design a pair of encoder and decoder so as to achieve bounded EEC under the individual subchannel power constraints:
\begin{equation}\label{eq:subchannelC}
		\lim_{k\rightarrow\infty}\mathbb{E}\{s_i(k)^2\}<\p_i\text{ for given }\p_i,\ i=1,\dots,l.
	\end{equation}

\item \textbf{Problem 2:} Design a pair of encoder and decoder so as to achieve bounded EEC under the total power constraints:
\begin{equation}\label{eq:totalC}
		\lim_{k\rightarrow\infty}\sum_{i=1}^{l}\mathbb{E}\{s_i(k)^2\}<\p\text{ for given }\p.
	\end{equation}
    \end{itemize}

For each problem, we shall first study the solvability condition, i.e., under what conditions an encoder-decoder pair can be constructed to fulfill the respective requirement? When the problems are solvable, we shall further study how to construct such desired encoder-decoder pair.

Before proceeding, we recall that the topological entropy of a discrete-time LTI source as in (\ref{eq:source}) is given by
\begin{align*}
\TE(A)=\sum_{i=1}^{\anti}\log |\lambda_i|,
\end{align*}
where $\lambda_1,\dots,\lambda_{\anti}$ denote the antistable eigenvalues of $A$. It quantifies the complexity, or the degree of instability of the source dynamics, in terms of the uncertainty growth rate.

\begin{remark}
    The noiseless channel feedback assumption is a theoretical simplification, which has been adopted in prior works on remote state estimation \cite{han2023coded,han2024coded,li2011optimal,jin2024remote}. It may serve as a reasonable approximation
in certain scenarios such as drone-based industrial monitoring, where the forward link for encoded signal transmission is often weak and noisy due to onboard battery limitation, while the feedback link supported by ground infrastructure allows a high SNR. 
\end{remark}

\section{Preliminaries}\label{sec:preliminary}
\subsection{Partially Ordered Sets}
We first review some basic knowledge of partially ordered sets \cite{davey2002introduction}.

A binary relation $\trianglelefteq$ on a set $\mathcal{S}$ is said to be a partial order if for all $x,y,z\in\mathcal{S}$,
\begin{enumerate}[i)]
	\item $x\trianglelefteq x$,\hfill  (reflexive)
	\item $x\trianglelefteq y$ and $y\trianglelefteq x$ imply $x=y$,\hfill (anti-symmetric)
	\item $x\trianglelefteq y$ and $y\trianglelefteq z$ imply $x\trianglelefteq z$.\hfill (transitive)
\end{enumerate}
A relation that is reflexive and transitive but not necessarily anti-symmetric is called a preorder.

A set $\mathcal{S}$ endowed with a partial order $\trianglelefteq$ forms a partially ordered set (poset), denoted by $(\mathcal{S},\trianglelefteq)$.
An element $x\in\mathcal{S}$ is called the least element of $\mathcal{S}$ if $x\trianglelefteq y$ for all $y\in\mathcal{S}$, and is callled the greatest element of $\mathcal{S}$ if $y\trianglelefteq x$ for all $y\in\mathcal{S}$. The least element and the greatest element, if exist, are unique. 
For a subset $\mathcal{V} \subseteq \mathcal{S}$, an element $x\in\mathcal{S}$ is called a lower bound of $\mathcal{V}$ if $x\trianglelefteq y$ for all $y\in\mathcal{V}$, and is called an upper bound of $\mathcal{V}$ if $y\trianglelefteq x$ for all $y\in\mathcal{V}$. 
The greatest element of the set of all lower bounds of $\mathcal{V}$, if exists, is called the infimum of $\mathcal{V}$. 
Similarly, the supremum of $\mathcal{V}$ is defined as the least element of the set of all its upper bounds.

A poset $(\mathcal{S},\trianglelefteq)$ is called a lattice if every pair of elements $x, y\in\mathcal{S}$ has its infimum and supremum, known as the meet and join of $x$ and $y$ (denoted by $x \land y$ and $x \lor y$), respectively.

\subsection{Majorization Theory}

Majorization theory serves as a fundamental tool in this work. We review some basic concepts and properties, and refer the readers to \cite{Marshall2011Majorization} for a comprehensive treatment.

For a vector $x=\begin{bmatrix}
	x_1&x_2&\cdots&x_n
\end{bmatrix}'\in\mathbb{R}^{n}$, denote by
\begin{equation*}
	x^\downarrow=\begin{bmatrix}x_1^\downarrow&x_2^\downarrow&\cdots&x_n^\downarrow\end{bmatrix}'
\end{equation*} 
its nonincreasing rearrangement, where $x_1^\downarrow\geq x_2^\downarrow\geq \cdots\geq x_n^\downarrow$, and by
\begin{equation*}
	x^\uparrow=\begin{bmatrix}x_1^\uparrow&x_2^\uparrow&\cdots&x_n^\uparrow\end{bmatrix}'
\end{equation*}
its nondecreasing rearrangement, where $x_1^\uparrow\leq x_2^\uparrow\leq \cdots\leq x_n^\uparrow$.

For two vectors $x,y\in\mathbb{R}^n$, $x$ is said to be majorized by $y$, denoted by $x\preccurlyeq y$, if
\begin{equation}\label{eq:maj2}
	\sum_{i=1}^k x_i^\downarrow\leq\sum_{i=1}^k y_i^\downarrow, k=1,\dots,n\!-\!1,\, \text{and }
	\sum_{i=1}^n x_i^\downarrow=\sum_{i=1}^n y_i^\downarrow,
\end{equation}
or equivalently,
\begin{equation}\label{eq:maj1}
	\sum_{i=1}^k x_i^\uparrow\geq\sum_{i=1}^k y_i^\uparrow, k=1,\dots,n\!-\!1,\, \text{and }
	\sum_{i=1}^n x_i^\uparrow=\sum_{i=1}^n y_i^\uparrow.
\end{equation}
If the equality in \eqref{eq:maj2} is replaced by inequality $\le$, $x$ is said to be weakly majorized by $y$ from below, denoted by $x\preccurlyeq_{\mathrm{w}}y$. If, further, the inequalities $\leq$ and equality in \eqref{eq:maj2} are all replaced by the strict inequality $<$, then $x$ is said to be strictly weakly majorized by $y$ from below, denoted by $x\prec_{\mathrm{w}}y$. Similarly, if the equality in \eqref{eq:maj1} is replaced by inequality $\geq$, $x$ is said to be weakly majorized by $y$ from above, denoted by $x\preccurlyeq^\mathrm{w} y$. 
If, further, the inequalities $\geq$ and equality in \eqref{eq:maj1} are all replaced by the strict inequality $>$, then $x$ is said to be strictly weakly majorized by $y$ from above, denoted by $x\prec^{\mathrm{w}}y$.

For vectors $x,y\in\mathbb{R}^n_{+}$, $x$ is said to be log-majorized by $y$, or equivalently, multiplicatively-majorized by $y$, denoted by $x\underset{\mathrm{log}}\preccurlyeq y$, if
\begin{equation}\label{eq:maj3}
	\prod_{i=1}^k x_i^\downarrow\leq\prod_{i=1}^k y_i^\downarrow, k=1,\dots,n\!-\!1,\, \text{and }
	\prod_{i=1}^n x_i^\downarrow=\prod_{i=1}^n y_i^\downarrow.
\end{equation}
If the inequalities $\leq$ and equality in \eqref{eq:maj3} are all replaced by the strict inequality $<$, $x$ is said to be strictly weakly log-majorized by $y$ from below, denoted by $x\underset{\log}{\prec_{\mathrm{w}}}y$.
When $x$ and $y$ are both positive vectors, $x\underset{\mathrm{log}}\preccurlyeq y$ is equivalent to $\log x\preccurlyeq \log y$, and $x\underset{\log}{\prec_{\mathrm{w}}} y$ is equivalent to $\log x\prec_{\mathrm{w}}\log y$.

Majorization orders the degree of fluctuation of the elements of vectors whose elements have the same average. 
Specifically, $x\preccurlyeq y$ indicates that the elements of $x$ have less fluctuation (are more evenly distributed)  than those of $y$.
Weak majorizations order the degree of fluctuation of elements and the average of elements in a combined way.

Majorization $\preccurlyeq$ and weak majorizations $\preccurlyeq_{\mathrm{w}}$ and $\preccurlyeq^{\mathrm{w}}$ are all preorders on $\mathbb{R}^n$. They become partial orders when restricted to the set
\begin{equation*}
    \mathcal{D}^n=\{x\in\mathbb{R}^n\mid x_1\geq x_2\geq\cdots\geq x_n\}.
\end{equation*} 
Furthermore, the posets $(\mathcal{D}^n, \preccurlyeq_{\mathrm{w}})$ and $(\mathcal{D}^n, \preccurlyeq^{\mathrm{w}})$ are lattices \cite{bapat1991majorization}. Taking $(\mathcal{D}^n, \preccurlyeq_{\mathrm{w}})$ as an example, Appendix~\ref{appendix:join} presents explicit ways to compute the meet and join of any two vectors $x,y\in\mathcal{D}^n$.

Now, we present several useful lemmas.

\begin{lma}[\!\!\cite{Marshall2011Majorization}]\label{lma:maj1}
	For vectors $x,y\in\mathbb{R}^n$, $x\preccurlyeq ^{\mathrm{w}} y$ ($x\prec^\mathrm{w} y$, resp.) if and only if there exists a $z\in\mathbb{R}^n$ such that $x\geq z$ ($x>z$, resp.) and $z\preccurlyeq y$.
\end{lma}

A procedure for constructing such a $z$ in Lemma \ref{lma:maj1} is given in \cite[5.A.9, 5.A.9.a]{Marshall2011Majorization}.

\begin{lma}\label{lma:maj2}
	For vectors $x\in\mathbb{R}^n_{+}$ and $y\in\mathbb{R}^l_{+}$, where $l\leq n$, $x\preccurlyeq_{\mathrm{w}}
		\begin{bmatrix}
			y\\
            \boldsymbol{0}
	\end{bmatrix}$ 
	if and only if there exists a $z\in\mathbb{R}^l_{+}$ such that $z\leq y$ and $x\preccurlyeq
		\begin{bmatrix}
			z\\
            \boldsymbol{0}
	\end{bmatrix}$.
\end{lma}

\begin{proof}
	The sufficiency is readily verified. We prove the necessity by constructing $z$. Without loss of generality, assume $y\in\mathcal{D}^l$. Denote $\delta=\sum_{i=1}^ly_i-\sum_{i=1}^nx_i\geq 0$. 
	Let $k$ be the largest integer in $\{1,\dots,l\}$ such that $\sum_{i=k}^ly_i\geq\delta$. Such a $k$ must exist since $\sum_{i=1}^ly_i\geq\delta$. 
	Construct
	\begin{align*}
		z_i=\begin{cases}
			y_i, &i=1,\dots,k-1;\\
			\sum\nolimits_{j=k}^ly_j-\delta, &i=k;\\
			0, &i=k+1,\dots,l.
		\end{cases}
	\end{align*} 
    Since $z_k=y_k+(\sum_{i=k+1}^ly_i-\delta)\leq y_k$, we have $0\leq z\leq y$. 
    Since $\sum_{i=1}^{t}z_i=\sum_{i=1}^{t}y_i\geq\sum_{i=1}^{t}x_i^\downarrow$ for $t=1,\dots,k-1$, and  $\sum_{i=1}^{k}z_i=\sum_{i=1}^{n}x_i\geq\sum_{i=1}^{t}x_i^\downarrow$ for $t=k,\dots,n-1$, 
    $x\preccurlyeq
		\begin{bmatrix}
			z\\
            \boldsymbol{0}
	\end{bmatrix}$ holds. 
    This completes the proof.
\end{proof}

\begin{lma}[\!\!\cite{Marshall2011Majorization}]\label{lma:symmetric_matrix}
	There exists a real symmetric matrix $X\in\mathbb{R}^{n\times n}$ with eigenvalues $\lambda_1,\lambda_2,\dots,\lambda_n$ and diagonal elements $d_1,d_2,\dots,d_n$ if and only if 
	\begin{equation*}
		\begin{bmatrix}
			d_1 & d_2 & \cdots& d_n
		\end{bmatrix}' \preccurlyeq \begin{bmatrix}
			\lambda_1 &  \lambda_2  &   \cdots&    \lambda_n
		\end{bmatrix}'.
	\end{equation*}
\end{lma}
\par
An approach for finding such an $X$ in Lemma \ref{lma:symmetric_matrix} is given in \cite[9.B.2]{Marshall2011Majorization}.

\begin{lma}[\!\!\cite{Marshall2011Majorization}]\label{lma:cyclic_matrix}
    There exists a matrix $X\in\mathbb{C}^{n\times n}$ with eigenvalues $\lambda_1, \lambda_2,\dots,\lambda_n$ and singular values $s_1, s_2, \dots, s_n$ if and only if 
    \begin{equation*}
    \begin{bmatrix}
			|\lambda_1| & |\lambda_2| & \cdots& |\lambda_n|
		\end{bmatrix}' 
        \underset{\mathrm{log}}\preccurlyeq \begin{bmatrix}
			s_1 &  s_2  &   \cdots&    s_n
		\end{bmatrix}'.
    \end{equation*}
\end{lma}

In Lemma~\ref{lma:cyclic_matrix}, if the complex eigenvalues appear in conjugate pairs, then a real matrix $X$ can be constructed \cite{li2001construction}.

Before proceeding, we introduce monotonic functions under majorization order. Let $\mathcal{S}\subseteq \mathbb{R}^n$. A function $f: \mathcal{S} \to \mathbb{R}$ is said to be Schur-convex on $\mathcal{S}$ if $x\preccurlyeq y$ implies $f(x)\le f(y)$ for all $x,y\in\mathcal{S}$. 
Also, $f$ is said to be Schur-concave on $\mathcal{S}$ if $-f$ is Schur-convex on $\mathcal{S}$.

\begin{lma}\label{lma:schurconvex}
	Let $f(x)=\sum_{i=1}^{n}a_ix_i$ with $0\leq a_1\leq a_2\leq \cdots\leq a_n$.
    Then $f$ is Schur-concave on $\mathcal{D}^n$.
    Furthermore, for $x,y\in\mathcal{D}^n$, if $x\prec^{\mathrm{w}}y$, then $f(x)> f(y)$.  
\end{lma}
\begin{proof}
	\!The proof follows readily from \cite[3.H.3, 3.H.3.a]{Marshall2011Majorization} and is omitted.
\end{proof}

\subsection{Submodular Functions}
Let $\mathcal{N}=\{1,2,\dots,n\}$ be an index set and $2^{\mathcal{N}}$ be its power set.
A set function $f:2^{\mathcal{N}}\rightarrow\mathbb{R}$ is said to be submodular if 
\begin{equation*}
    f(\mathcal{S})+f(\mathcal{T})\geq f(\mathcal{S}\cup\mathcal{T})+f(\mathcal{S}\cap\mathcal{T}),\ \forall\,\mathcal{S},\mathcal{T}\subseteq\mathcal{N},
\end{equation*}
where $f(\emptyset)=0$.
Denote the submodular polyhedron and base polyhedron associated with $f$ by $\mathcal{P}_f$ and $\mathcal{B}_f$, given by 
\begin{equation*}
\begin{aligned}
        \mathcal{P}_f&=\Big\{x\in\mathbb{R}^n\:\Big|\:\sum\nolimits_{i\in\mathcal{S}}x_i\leq f(\mathcal{S}),\;\forall\mathcal{S}\subseteq \mathcal{N}\Big\},\\
        \mathcal{B}_f&=\Big\{x\in\mathcal{P}_f\:\Big|\:\sum\nolimits_{i\in\mathcal{N}}x_i=f(\mathcal{N})\Big\}.
\end{aligned}
\end{equation*}

The following lemma is useful, which is adapted from \cite[Lemma~9]{schoot2025characterization}. The proof is omitted for brevity.
\begin{lma}\label{lma:base_poly}
    Let $f:2^{\mathcal{N}}\rightarrow\mathbb{R}$ be a submodular function.
    For each $b\in\mathbb{R}^n$, there exists $x^\star\in\mathcal{B}_f$ such that $x^\star+b\preccurlyeq x+b$ for all $x\in\mathcal{B}_f$.
\end{lma}

For more knowledge on submodular functions and associated theory, we refer the readers to \cite{fujishige2005submodular}.

\section{Linear Coding With Subchannel Power Constraints}\label{sec:subchannel}
\subsection{Main Result}

The following theorem gives a necessary and sufficient condition for the solvability of Problem 1. Recall that $\lambda_1,\dots,\lambda_{\anti}$ denote the antistable eigenvalues of $A$. Let 
\begin{align*}
\bar{l}=\min\{l,\anti\}.
\end{align*}

\begin{theorem}\label{thm:subchannel}
    Problem 1 is solvable if and only if there exist $\gamma_i\geq1$, $i=1,\dots,{\bar{l}}$, such that 
\begin{subequations}\label{eq:cor_opt_pro}
			\begin{align}
                &\begin{bmatrix}
						\frac{\p_1}{\sigma_1^2}+1&\frac{\p_2}{\sigma_2^2}+1&\cdots&\frac{\p_{l}}{\sigma_{l}^2}+1
					\end{bmatrix}'\prec^{\mathrm{w}}\notag\\ 
					&\qquad\qquad\qquad\quad \begin{bmatrix}
						\gamma_1&\gamma_2&\cdots&\gamma_{\bar{l}}&1 &\cdots&1
					\end{bmatrix}',\label{eq:cor_opt_P_sigma}\\
					&\begin{bmatrix}
						|\lambda_1|^2&|\lambda_2|^2&\cdots&|\lambda_{\anti}|^2
					\end{bmatrix}' \underset{\mathrm{log}}\preccurlyeq\notag\\ 
					&\qquad\qquad\qquad\quad \begin{bmatrix}
						\gamma_1&\gamma_2&\cdots&\gamma_{\bar{l}}&1 &\cdots&1
					\end{bmatrix}'.\label{eq:cor_opt_gamma}
			\end{align}
		\end{subequations}
\end{theorem}

The proof of Theorem~\ref{thm:subchannel}, along with the detailed encoder-decoder design procedure, is postponed to Subsection~\ref{sec:Proof_of_theorem}. 

The solvability condition given in Theorem~\ref{thm:subchannel} involves seeking a collection of auxiliary variables $\gamma_i\geq 1,i=1,\dots,\bar{l}$, and checking two majorization inequalities (\ref{eq:cor_opt_P_sigma}) and (\ref{eq:cor_opt_gamma}), one in additive form and the other in multiplicative form, coupled through the auxiliary variables. The two majorization inequalities indicate that a more even distribution of SNRs $\frac{p_1}{\sigma_1^2},\dots,\frac{p_l}{\sigma_l^2}$ across the subchannels, or a more even distribution of the log-magnitudes of the antistable poles $\log\! |\lambda_1|,\dots,\log\! |\lambda_{n_a}|$ of the source, facilitates the solvability of the problem. 

We note that verifying the solvability condition is essentially convex. Since majorization is permutation-invariant, we can assume without loss of generality that $\gamma_i,i=1,\dots,\bar{l}$ are in nonincreasing order. Then the solvability condition reduces to a geometric program with variables $\gamma_1,\dots,\gamma_{\bar{l}}$ \cite{boyd2004convex}, which can be readily transformed into a convex form. 

Denote by $\mathcal{P}_{\lambda}$ the set of power constraints $\begin{bmatrix}p_1&\cdots&p_l\end{bmatrix}'$ rendering Problem 1 solvable, called the feasible power region, for an unstable 
source with antistable poles
\[
\lambda=\begin{bmatrix}\lambda_1&\cdots&\lambda_{n_a}\end{bmatrix}'.
\]
Then the following two propositions hold.

\begin{proposition}\label{Plambda}
    $\mathcal{P}_{\lambda}$ is a convex set.
\end{proposition}

\begin{proof}
    Let $\bar{p}=\begin{bmatrix}\bar{p}_1&\cdots&\bar{p}_l\end{bmatrix}'$ and $\hat{p}=\begin{bmatrix}\hat{p}_1&\cdots&\hat{p}_l\end{bmatrix}'$ be arbitrary elements in $\mathcal{P}_{\lambda}$. 
    Then, there exist $\bar{\gamma}_i\geq1,i=1,\dots,\bar{l}$ and $\hat{\gamma}_i\geq1,i=1,\dots,\bar{l}$ such that \eqref{eq:cor_opt_pro} holds under the power constraints $\bar{p}$ and $\hat{p}$, respectively.
    Without loss of generality, we assume that $\bar{\gamma}_1\geq\cdots\geq \bar{\gamma}_{\bar{l}}$ and $\hat{\gamma}_1\geq\cdots\geq \hat{\gamma}_{\bar{l}}$. Let $\tilde{p}=\theta \bar{p} +(1-\theta)\hat{p}$ for any $\theta \in [0, 1]$.
    To prove $\tilde{p}\in\mathcal{P}_\lambda$,
    we construct $\tilde{\gamma}_i=\bar{\gamma}_i^\theta\hat{\gamma}_i^{1-\theta}$, which satisfies $\tilde{\gamma}_1\geq\cdots\geq\tilde{\gamma}_{\bar{l}}\geq1$.
    We observe that for $k=1,\dots,\bar{l}-1$, 
    \begin{equation*}
        \begin{aligned}
            \prod_{i=1}^{k}\tilde{\gamma}_i&=\bigg(\prod_{i=1}^{k}\bar{\gamma}_i\bigg)^\theta\bigg(\prod_{i=1}^{k}\hat{\gamma}_i\bigg)^{1-\theta}\geq\bigg(\prod_{i=1}^{k}|\lambda_i|^{\downarrow}\bigg)^{2},\text{ and}\\
            \prod_{i=1}^{\bar{l}}\tilde{\gamma}_i&=\bigg(\prod_{i=1}^{\bar{l}}\bar{\gamma}_i\bigg)^\theta\bigg(\prod_{i=1}^{\bar{l}}\hat{\gamma}_i\bigg)^{1-\theta}=\prod_{i=1}^{n_a}|\lambda_i|^{2},
        \end{aligned}
    \end{equation*}
 which implies that $\tilde{\gamma}_1,\dots,\tilde{\gamma}_{\bar{l}}$ satisfy \eqref{eq:cor_opt_gamma}. Furthermore, for $k=1,\dots,\bar{l}$, we have
    \begin{equation*}
        \begin{aligned}
            \sum_{i=k}^l\!\bigg(\frac{\tilde{p}_{i}}{\sigma_i^2}+1\bigg)^{\!\downarrow}
             &\!\!\geq\!\theta\sum_{i=k}^l\bigg(\frac{\bar{p}_{i}}{\sigma_i^2}+1\bigg)^{\!\downarrow}+(1-\theta)\sum_{i=k}^l\bigg(\frac{\hat{p}_{i}}{\sigma_i^2}+1\bigg)^{\!\downarrow}\\
             &\!\!>\!\theta\bigg(\!\sum_{i=k}^{\bar{l}}\!\bar{\gamma}_i+l-\bar{l}\bigg)+(1\!-\!\theta)\bigg(\!\sum_{i=k}^{\bar{l}}\!\hat{\gamma}_i+l-\bar{l}\bigg)\\
             &\!\!=\!\sum_{i=k}^{\bar{l}}(\theta\bar{\gamma}_i+(1-\theta)\hat{\gamma}_i)+l-\bar{l}\\
            &\!\!\geq\! \sum_{i=k}^{\bar{l}}\bar{\gamma}_i^\theta\hat{\gamma}_i^{1-\theta}+l-\bar{l}=\sum_{i=k}^{\bar{l}}\tilde{\gamma}_i+l-\bar{l},
        \end{aligned}
    \end{equation*}
    where the first inequality is due to the concavity of the function $(x_1,\dots,x_l)\mapsto\sum_{i=k}^lx_i^\downarrow$, and the last inequality follows from the weighted arithmetic-geometric mean inequality.
    Thus, $\tilde{\gamma}_1,\dots,\tilde{\gamma}_{\bar{l}}$ also satisfy \eqref{eq:cor_opt_P_sigma} under power constraints $\tilde{p}$, establishing that $\tilde{p}\in\mathcal{P}_\lambda$. This completes the proof.
\end{proof}

\begin{proposition}\label{fpr}
Consider two sources with antistable poles represented by $\lambda$ and $\bar{\lambda}$ respectively and a set of subchannels with given noise variances. If $\log\!|\lambda| \preccurlyeq\log\!|\bar{\lambda}|$, then $\mathcal{P}_{\lambda}\supseteq\mathcal{P}_{\bar{\lambda}}$.
\end{proposition}

\begin{proof}
    For any $\bar{p}\in\mathcal{P}_{\bar{\lambda}}$, there exist $\bar{\gamma}_i\geq1,i=1,\dots,\bar{l}$ such that the majorization inequalities \eqref{eq:cor_opt_pro} corresponding to $\bar{\lambda}$ and $\bar{p}$ hold. 
    Since $\log\!|\lambda| \preccurlyeq \log\!|\bar{\lambda}|$, the transitivity of the majorization order implies that these $\bar{\gamma}_i$  also satisfy \eqref{eq:cor_opt_pro} for $\lambda$ and $\bar{p}$. Therefore, $\bar{p}\in\mathcal{P}_{\lambda}$, which completes the proof.
\end{proof}

Proposition~\ref{fpr} demonstrates that the difficulty of transmitting a source via LTI coding is intrinsically constrained not only by the total sum of the log-magnitudes of the antistable poles, namely, the topological entropy, but also by how evenly these log-magnitudes are distributed. A more even distribution of $\log |\lambda_i|,i=1,\dots,\anti,$ is more favorable for transmission. An example illustrating Proposition~\ref{fpr} is as follows.

\begin{exam}\label{example1}
Consider two sources with antistable poles 
\begin{align*}
\lambda=\begin{bmatrix} e^{0.35}&e^{0.3}&e^{0.25} \end{bmatrix}^\prime \text{ and } \bar{\lambda}=\begin{bmatrix}e^{0.8}&e^{0.05}&e^{0.05}\end{bmatrix}^\prime,
\end{align*}
respectively. One can verify that $\log|\lambda|\preccurlyeq\log|\bar{\lambda}|$. Consider a parallel Gaussian channel comprising two subchannels with noise variances given by $\sigma_1^2 = 2, \sigma_2^2 = 3$. The feasible power regions for the two sources are plotted in Fig.~\ref{fig:feasible_region}. 
One can see that the feasible power region for source corresponding to $\lambda$ contains that for source corresponding to $\bar{\lambda}$.
\end{exam}

\begin{figure}[htbp]
	\centering
	\includegraphics[width=0.9\linewidth]{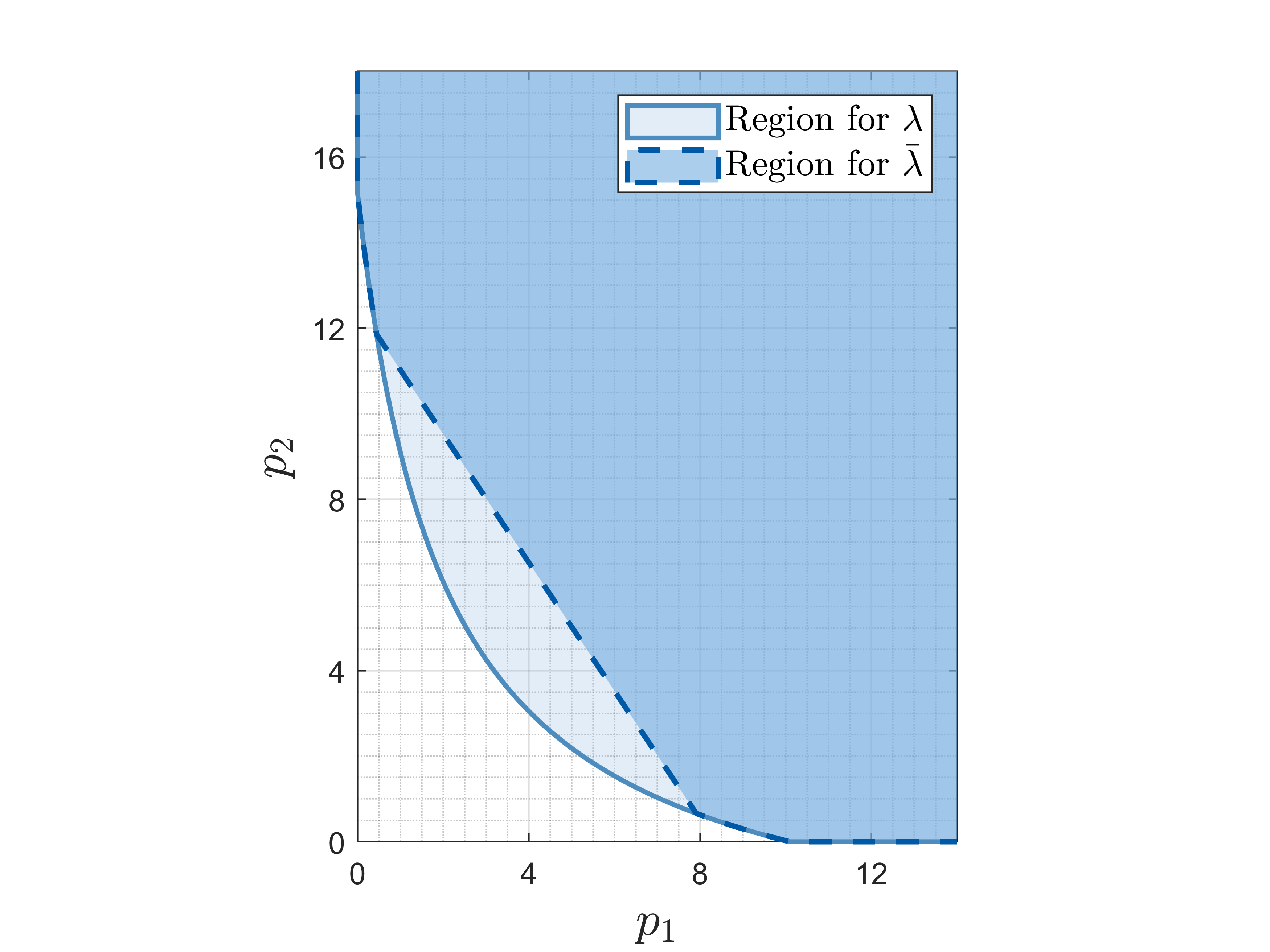}
	\caption{Comparison of feasible power regions for two different sources in Example \ref{example1}. 
    }
	\label{fig:feasible_region}
\end{figure}

The following corollaries follow from Theorem~\ref{thm:subchannel}. They give sufficient solvability conditions directly in terms of majorization inequalities between subchannel SNRs and magnitudes of the antistable poles.

\begin{corollary}\label{corollary1}
      When $l\geq n_a$, Problem 1 is solvable if
    \begin{align}
    &\begin{bmatrix}
    \frac{\p_1}{\sigma_1^2}+1&\frac{\p_2}{\sigma_2^2}+1&\cdots&\frac{\p_{l}}{\sigma_{l}^2}+1
\end{bmatrix}'\prec^{\mathrm{w}}\notag\\ 
&\qquad\begin{bmatrix}
    |\lambda_1|^2&|\lambda_2|^2&\cdots&|\lambda_{n_a}|^2&1 &\cdots&1
\end{bmatrix}'.\label{eq:snr_pole_cor1}
    \end{align}
\end{corollary}

\begin{proof}
When \eqref{eq:snr_pole_cor1} holds, setting
$\gamma_i \!=\! |\lambda_i|^2$, $i=1,\ldots,n_a$
satisfies the solvability condition \eqref{eq:cor_opt_pro} of Theorem~\ref{thm:subchannel}.
\end{proof}

\begin{corollary}\label{corollary2}
    When $l\leq n_a$, Problem 1 is solvable if
    \begin{align}
    &\begin{bmatrix}|\lambda_1|^2&|\lambda_2|^2&\cdots&|\lambda_{n_a}|^2
\end{bmatrix}'\underset{\log}{\prec_{\mathrm{w}}}\notag\\ 
&\qquad\begin{bmatrix}
    \frac{\p_1}{\sigma_1^2}+1&\frac{\p_2}{\sigma_2^2}+1&\cdots&\frac{\p_{l}}{\sigma_{l}^2}+1&1&\cdots&1
\end{bmatrix}'.\label{eq:pole_snr_cor2}
    \end{align}
\end{corollary}

\begin{proof}
By \eqref{eq:pole_snr_cor2}, there exists a constant $\epsilon<1$ such that $\epsilon\big(\frac{\p_i}{\sigma_i^2}+1\big)>1$ for all $i=1,\dots,l$ and 
    \begin{align*}
    &\begin{bmatrix}|\lambda_1|^2&|\lambda_2|^2&\cdots&|\lambda_{n_a}|^2
\end{bmatrix}'\underset{\log}{\prec_{\mathrm{w}}}\notag\\ 
&\begin{bmatrix}
    \epsilon\big(\frac{\p_1}{\sigma_1^2}+1\big)&\epsilon\big(\frac{\p_2}{\sigma_2^2}+1\big)&\cdots&\epsilon\big(\frac{\p_{l}}{\sigma_{l}^2}+1\big)&1&\cdots&1
\end{bmatrix}'.
    \end{align*}
Then by Lemma~\ref{lma:maj2}, there exist $\gamma_i$, $i=1,\dots,l$ that satisfy $1\leq\gamma_i\leq\epsilon\big(\frac{p_i}{\sigma_i^2}+1\big)$ and \eqref{eq:cor_opt_gamma}. 
Furthermore, since $\epsilon<1$, we have $\gamma_i<\frac{p_i}{\sigma_i^2}+1$, implying that \eqref{eq:cor_opt_P_sigma} also holds. Thus, by Theorem~\ref{thm:subchannel}, Problem~1 is solvable.
\end{proof}

\begin{remark}
    Recall that the capacity of the $i$th subchannel is given by \cite{Thomas2006Elements}
\begin{equation*}
	\ca_i=\frac{1}{2}\log\left(1+\frac{\p_i}{\sigma_i^2}\right),\quad i=1,\dots,l.
\end{equation*}
The sufficient condition (\ref{eq:pole_snr_cor2}) can then be rewritten as
\begin{multline*}
\begin{bmatrix}\log|\lambda_1|&\log|\lambda_2|&\dots&\log|\lambda_{n_a}|\end{bmatrix}'\prec_{\mathrm{w}}\\
\begin{bmatrix}c_1&c_2&\dots&c_{l}&0&\dots&0\end{bmatrix}',
\end{multline*}
which is consistent with the condition in \cite[Theorem 2.1]{zaidi2016tightness}. In this sense, Theorem \ref{thm:subchannel} sharpens the corresponding result in \cite{zaidi2016tightness} by giving an exact characterization of solvability under prescribed subchannel powers.
\end{remark}

\begin{remark}
While the conditions in the above corollaries are simpler to verify, they are generally conservative. Corollary~\ref{corollary1} becomes necessary when $n_a=1$, namely, when the source has only one single antistable pole, whereas Corollary~\ref{corollary2} becomes necessary when $l=1$, namely, when there is only one single subchannel.
\end{remark}

\subsection{Solvability and Channel Capacity}

In view of Theorem~\ref{thm:subchannel}, the solvability of Problem 1 amounts to the feasibility of \eqref{eq:cor_opt_pro} over the variables $\gamma_i\geq 1,i=1,\dots,\bar{l}$.
Applying Lemma~\ref{lma:maj1} to (\ref{eq:cor_opt_P_sigma}) and invoking the Schur-concavity of the function $(x_1,\dots,x_l)\mapsto\sum_{i=1}^l \log x_i$ \cite{Marshall2011Majorization} yields
\begin{align*}
\ctotal=\sum_{i=1}^l \log \left(1+\frac{p_i}{\sigma_i^2}\right)> \sum_{i=1}^{\bar{l}} \log \gamma_i.    
\end{align*}
This, combined with (\ref{eq:cor_opt_gamma}), yields $\ctotal>\TE(A)$.
This means that for Problem 1 to be solvable, the total channel capacity must exceed the topological entropy of the source. However, such a total capacity requirement is in general insufficient.

\begin{exam}
Recall the source with antistable poles $\bar{\lambda}$ and the parallel Gaussian channel described in Example~\ref{example1}. Suppose the power constraints are set to $p_1=4$ and $p_2=4$. 
Although the total channel capacity $\ctotal=0.973$ exceeds the source's topological entropy $\TE(A)=0.9$, Problem~1 is infeasible under this setting since $(4,4)$ lies outside the feasible power region for $\bar{\lambda}$, as illustrated in Fig.~\ref{fig:feasible_region}.
\end{exam}

A natural question then arises: Does there exist a class of sources for which the requirement $\ctotal>\TE(A)$ is not only necessary but also sufficient for the solvability of Problem~1? If so, what special features do such sources possess, and how can they be characterized? These questions are answered by the following theorem.

\begin{theorem}\label{cor:C>H}
    The condition $\ctotal>\TE(A)$ is necessary and sufficient for solvability of Problem 1 if and only if
\begin{equation}\label{eq:source_even}
      \log|\lambda| \preccurlyeq {\textstyle \frac{\TE(A)}{l}} \begin{bmatrix}
           \boldsymbol{1}_l\\
            \boldsymbol{0}
                 \end{bmatrix}.
    \end{equation}
\end{theorem}

\begin{proof}
    \!We first prove sufficiency. Suppose that \eqref{eq:source_even} holds. 
    We show that Problem~1 is solvable whenever $\ctotal>\TE(A)$.
    By $\ctotal>\TE(A)$, there exists $\{\gamma_i\}_{i=1}^l$ such that  $1\leq \gamma_i< \frac{\p_i}{\sigma_i^2}+1$ and   $\sum_{i=1}^l\frac{1}{2}\log\!\gamma_i=\TE(A)$. 
    Then \eqref{eq:cor_opt_P_sigma} is satisfied. 
    By \eqref{eq:source_even}, $n_a\geq l$. Together with $\frac{\TE(A)}{l}\boldsymbol{1}_l\!\preccurlyeq\! \begin{bmatrix}
         \frac{1}{2}\log\!\gamma_1\!&\!\cdots\!&\!\frac{1}{2}\log\!\gamma_l
     \end{bmatrix}^\prime$, the transitivity of majorization order yields \eqref{eq:cor_opt_gamma}.
 Thus, by Theorem~1, Problem~1 is solvable.

    We next prove necessity. Suppose that Problem~1 is solvable whenever $\ctotal\!>\!\TE(A)$. We show  \eqref{eq:source_even} holds. For any $\epsilon>1$, consider the subchannel capacities $\ca_{1\epsilon}=\cdots=\ca_{l\epsilon}
=
\frac{\TE(A)}{l}+\frac{1}{2}\log\epsilon$. Since $\sum_{i=1}^lc_{i\epsilon}>\TE(A)$, Problem~1 is solvable. By Theorem~\ref{thm:subchannel}, there exist $\gamma_{i\epsilon}\geq1$, $i=1,\dots,\bar{l}$ satisfying \eqref{eq:cor_opt_pro}. If $\bar l<l$, additionally define $\gamma_{i\epsilon}=1$, $i=\bar l+1,\dots,l$. From \eqref{eq:cor_opt_P_sigma},
\begin{equation}\label{eq:gamma_epsilon}
    \sum\nolimits_{i=1}^{l}\gamma_{i\epsilon}
<
\sum\nolimits_{i=1}^le^{2\ca_{i\epsilon}}=\epsilon l e^{\frac{2}{l}\TE(A)}.
\end{equation}
By \eqref{eq:cor_opt_gamma} and arithmetic-geometric mean inequality, we obtain
\begin{equation}\label{eq:AM_gamma}
    \! \!\!\frac{1}{l}\!\sum\nolimits_{i=1}^{{l}}\gamma_{i\epsilon}\geq\Big(\!\prod\nolimits_{i=1}^{{l}}\gamma_{i\epsilon}\Big)^{\frac{1}{l}}=\prod\nolimits_{i=1}^{\anti}|\lambda_i|^{\frac{2}{l}}=e^{\frac{2}{l}\TE(A)}.
\end{equation}
Combining this with \eqref{eq:gamma_epsilon} gives $\lim_{\epsilon\rightarrow1}\frac{1}{l}\sum_{i=1}^{{l}}\gamma_{i\epsilon}=e^{\frac{2}{l}\TE(A)}$.
Since the geometric mean remains equal to
$e^{\frac{2}{l}\TE(A)}$, equality in \eqref{eq:AM_gamma} is approached as $\epsilon\rightarrow1$, which implies $\lim_{\epsilon\rightarrow1}\gamma_{i\epsilon}
=e^{\frac{2}{l}\TE(A)}$, $i=1,\dots,l$.
Since $e^{\frac{2}{l}\TE(A)}>1$, this further implies $\bar l=l$. Finally, since $\{\gamma_{i\epsilon}\}_{i=1}^l$ satisfies \eqref{eq:cor_opt_gamma} for any $\epsilon>1$, taking the limit $\epsilon\rightarrow1$ in this majorization relation yields \eqref{eq:source_even}. This completes the proof.
\end{proof}

Theorem \ref{cor:C>H} shows that the total capacity condition $\ctotal\!>\!\TE(A)$ is sufficient only when the magnitudes of the antistable poles are sufficiently evenly distributed. Note that the inequality (\ref{eq:source_even}) automatically holds when all antistable poles have the same magnitude. It also holds when $l=1$, i.e., there is only a single subchannel, thereby recovering the classical data-rate condition for transmission over a single channel.

\subsection{Proof of Theorem~\ref{thm:subchannel} and Encoder-decoder Design}\label{sec:Proof_of_theorem}
\begin{myproof}[Theorem~\ref{thm:subchannel}]
We first prove the necessity. Suppose that an encoder \eqref{eq:encoder} and a decoder \eqref{eq:decoder} are designed such that the EEC is bounded and the subchannel power constraints are satisfied. 
Without loss of generality, we assume \eqref{eq:encoder} is minimal.
Denote by $\tf{T}_v(z)$ and $\tf{T}_w(z)$ the transfer matrices from $v$ and $w$ to $s$, respectively.
By \eqref{eq:subchannelC}, $\tf{T}_v(z)$ and $\tf{T}_w(z)$ must be stable, and we have
	\begin{align}
		\p_i&>\frac{1}{2\pi}\int_{-\pi}^{\pi}\!\{{\tf{T}_v(e^{j\omega})V\tf{T}_v(e^{j\omega})^*\!+\!\tf{T}_w(e^{j\omega})\Sigma\tf{T}_w(e^{j\omega})^*}\}_{ii}\,\mathrm{d}\omega\notag\\
		&\ge \frac{1}{2\pi}\int_{-\pi}^{\pi}\{\tf{T}_w(e^{j\omega})\Sigma\tf{T}_w(e^{j\omega})^*\}_{ii}\,\mathrm{d}\omega\notag\\
		&=\frac{1}{2\pi}\int_{-\pi}^{\pi}\{\Sigma^{\frac{1}{2}}\Sigma^{-\frac{1}{2}}\tf{T}_w(e^{j\omega})\Sigma\tf{T}_w(e^{j\omega})^*\Sigma^{-\frac{1}{2}}\Sigma^{\frac{1}{2}}\}_{ii}\,\mathrm{d}\omega\notag\\
		&=\sigma_i^2\frac{1}{2\pi}\int_{-\pi}^{\pi}\{\hat{\tf{T}}(e^{j\omega})\hat{\tf{T}}(e^{j\omega})^*\}_{ii}\,\mathrm{d}\omega,\label{eq:thm1P}
	\end{align}
	where $\hat{\tf{T}}(z)=\Sigma^{-\frac{1}{2}}\tf{T}_w(z)\Sigma^{\frac{1}{2}}$. 
    Direct calculation gives
    \begin{equation*}
        {\tf{T}}_w(z)=\begin{bmatrix}
            M&R
        \end{bmatrix}\begin{bmatrix}
            zI-A_e &-L\\
            -M&zI-R
        \end{bmatrix}^{-1}\begin{bmatrix}
            0\\
            I
        \end{bmatrix}.
    \end{equation*}
    Then we obtain
    \begin{equation*}
        \hat{\tf{T}}(z)=\hat{M}(zI-\hat{A}-\hat{L}\hat{M})^{-1}\hat{L},
    \end{equation*}
    where $\hat{M}=\Sigma^{-\frac{1}{2}}\begin{bmatrix}
        M&R
    \end{bmatrix}$, $\hat{A}=\begin{bmatrix}
        A_e&L\\
        0&0
    \end{bmatrix}$, and $\hat{L}=\begin{bmatrix}
        0\\
        I
    \end{bmatrix}\Sigma^{\frac{1}{2}}$. 
    By Lemma~\ref{lma:necessary_con} (see Appendix~\ref{appendix:necessary_con}), $\hat{A}+\hat{L}\hat{M}$ is Schur stable.
	Then, it follows from $\mathcal{H}_2$ optimal control theory \cite{zhou1996robust} that \begin{equation*}\label{eq:thmlmaCXC}
		\frac{1}{2\pi}\int_{-\pi}^{\pi}\hat{\tf{T}}(e^{j\omega})\hat{\tf{T}}(e^{j\omega})^*\,\mathrm{d}\omega\geq\hat{M}\hat{X}\hat{M}',
	\end{equation*}
	where $\hat{X}$ is the semi-stabilizing solution to the following discrete algebraic Riccati equation (DARE)
	\begin{equation}\label{eq:thm_dare}
		\hat{X}=\hat{A}\hat{X}(I+\hat{M}'\hat{M}\hat{X})^{-1}\hat{A}'.
	\end{equation}
	Together with \eqref{eq:thm1P}, we obtain
	\begin{equation}\label{eq:thmp_sigma}
		\frac{\p_i}{\sigma_i^2}>\frac{1}{2\pi}\int_{-\pi}^{\pi}\{\hat{\tf{T}}(e^{j\omega})\hat{\tf{T}}(e^{j\omega})^*\}_{ii}\,\mathrm{d}\omega\geq\{\hat{M}\hat{X}\hat{M}'\}_{ii}.
	\end{equation}
	Let $\mu_i$, $i=1,\dots,l$, denote the eigenvalues of $\hat{M}\hat{X}\hat{M}'$ ordered nonincreasingly, i.e.,
	\begin{equation}\label{eq:thmmu}
		\mu_1\geq\mu_2\geq\cdots\geq\mu_l\ge0.
	\end{equation}
	Then it follows from Lemma~\ref{lma:symmetric_matrix} that
	\begin{equation*}\label{eq:thmCXCdlambda}
		\begin{bmatrix}
			\{\hat{M}\hat{X}\hat{M}'\}_{11}&\cdots&\{\hat{M}\hat{X}\hat{M}'\}_{ll}
		\end{bmatrix}'\preccurlyeq\begin{bmatrix}
			\mu_1&\cdots&\mu_l
		\end{bmatrix}'.
	\end{equation*}
	Combining this with \eqref{eq:thmp_sigma} yields
	\begin{equation}\label{eq:thmweaklymajor}
		\begin{bmatrix}
			\frac{\p_1}{\sigma_1^2}&\frac{\p_2}{\sigma_2^2}&\cdots&\frac{\p_{l}}{\sigma_{l}^2}
		\end{bmatrix}'\prec^{\mathrm{w}}\begin{bmatrix}
			\mu_1&\mu_2&\cdots&\mu_{l}
		\end{bmatrix}'.
	\end{equation}
    By Lemma~\ref{lma:necessary_con},
    all unstable eigenvalues of $A$ are contained within those of $A_e$, and hence within those of $\hat{A}$.
    Suppose $\hat{A}$ has $m$ antistable eigenvalues in total, denoted $\lambda_1,\lambda_2,\dots,\lambda_{m}$, with the first $\anti$ being the antistable eigenvalues of $A$.
	 Then, applying Lemma~\ref{lma:major_inequality} (see Appendix~\ref{appdix:key_lemma}) to the DARE \eqref{eq:thm_dare} yields
\begin{equation*}
	\begin{aligned}
	&\begin{bmatrix}
	\log|\lambda_1|^2&\!\cdots\!&\log|\lambda_{\anti}|^2&\log|\lambda_{\anti+1}|^2&\!\cdots\!&\log|\lambda_{m}|^2
	\end{bmatrix}' \\ &\quad\ \ \preccurlyeq\begin{bmatrix}
	\log(1+\mu_1)&\!\cdots\!&\log(1+\mu_{\min\{l,m\}})&0&\!\cdots\!&0
	\end{bmatrix}'.
	\end{aligned}
	\end{equation*}
    It then follows from \eqref{eq:thmmu} that
    \begin{equation*}\label{eq:thm1majorization}
	\begin{aligned}
		&\begin{bmatrix}
			\log|\lambda_1|^2&\log|\lambda_2|^2&\cdots&\log|\lambda_{\anti}|^2
		\end{bmatrix}' {\preccurlyeq_{\mathrm{w}}}\\
		& \begin{bmatrix}
			\log(1+\mu_1)&\cdots&\log(1+\mu_{\bar{l}}) &0 &\cdots&0
		\end{bmatrix}'.
	\end{aligned}
	\end{equation*}
    Hence, by Lemma~\ref{lma:maj2}, there exist $1\leq\gamma_i\leq 1+\mu_i$, $i=1,\dots,\bar{l}$
	such that \eqref{eq:cor_opt_gamma} is satisfied. From \eqref{eq:thmweaklymajor}, it follows that \eqref{eq:cor_opt_P_sigma} is also satisfied, which completes the proof of necessity.
	
	We next prove the sufficiency. Suppose that there exist $\gamma_i\geq1,i=1,\dots,\bar{l}$ satisfying \eqref{eq:cor_opt_pro}.
	We shall construct an encoder and a decoder such that the EEC is bounded and the subchannel power constraints are met.

        For brevity, assume that all eigenvalues of $A$ lie outside the unit circle, i.e., $\anti=n$. This assumption can be removed using arguments similar to those in \cite{braslavsky:2007feedback,qiu2013stabilization}.
	We first consider the case when $A$ has distinct eigenvalues.
	It follows from \eqref{eq:cor_opt_gamma} and Lemma~\ref{lma:cyclic_matrix} that a real matrix $Y$ with eigenvalues $\lambda_1,\dots,\lambda_{n}$ and singular values $\gamma_1^{\frac{1}{2}},\dots,\gamma_{\bar{l}}^{\frac{1}{2}},1,\dots,1$ can be constructed. 
    Since $Y$ and $A$ have the same Jordan canonical form, they are similar, meaning that there is a real nonsingular matrix $Z$ such that $Y=Z^{-1}AZ$.
    Applying the singular value decomposition to $Y$ yields
	\begin{equation*}
		Y=P\begin{bmatrix}
			\Gamma&0\\
			0&I
		\end{bmatrix}^\frac{1}{2}Q',
	\end{equation*}
	where $\Gamma=\diag\{\gamma_1,\dots,\gamma_{\bar{l}}\}$, and $P,Q\in\mathbb{R}^{n\times n}$ are orthogonal matrices. Let 
	\begin{equation}\label{eq:thmMu}
	\tilde{M}=\begin{bmatrix}
	    (\Gamma-I)^{\frac{1}{2}}&0_{\bar{l}\times (n-\bar{l})}
	\end{bmatrix}Q'Z^{-1}.
	\end{equation}
	Then, we obtain
	\begin{align}
		&Z^{-1}AZ(I+Z'\tilde{M}'\tilde{M}Z)^{-1}Z'A'(Z')^{-1}\notag\\
		&=P\!\begin{bmatrix}
			\Gamma&0\\
			0&I
		\end{bmatrix}^\frac{1}{2}\!Q'\!\left(\!I\!+\!Q\begin{bmatrix}
		\Gamma-I&0\\
		0&0
		\end{bmatrix}Q'\!\right)^{-\!1}\!Q\!\begin{bmatrix}
		\Gamma&0\\
		0&I
		\end{bmatrix}^\frac{1}{2}\!P'\notag\\
		&=I,\label{eq:derivation_DARE}
	\end{align}
	which implies that $X:=ZZ'$ satisfies the following DARE
	\begin{equation}\label{eq:thmX_u}
		X=A(I+X\tilde{M}'\tilde{M})^{-1}XA'.
	\end{equation}
    Let 	\begin{equation}\label{eq:thmLu}
		\tilde{L}=-A(I+X\tilde{M}'\tilde{M})^{-1}X\tilde{M}'.
	\end{equation}
    Since the matrix $A+\tilde{L}\tilde{M}\!=\!A(I+X\tilde{M}'\tilde{M})^{-1}\!=\!X(A')^{-1}X^{-1}$
    is Schur stable, $X$ is in fact the stabilizing solution to \eqref{eq:thmX_u}.
    Now, design the parameter matrices of the encoder \eqref{eq:encoder} as
    \begin{equation}\label{eq:design_encoder}
        \begin{aligned}
            \left[\begin{array}{ccc}
		A_{\en} & K & L\\
		M & N & R 
	\end{array}\right]=\left[\begin{array}{ccc}
		A & 0 & A\bar{L}\\
		\bar{M} & -\bar{M} & \bar{M}\bar{L} 
	\end{array}\right],
        \end{aligned}
    \end{equation}
	where
	\begin{align}
		\bar{M}=\epsilon\Sigma^{\frac{1}{2}}U\tilde{M} \ \text{ and }\  \bar{L}=\frac{1}{\epsilon}\tilde{L}U'\Sigma^{-\frac{1}{2}},\label{eq:thm_encoderdecoder}
	\end{align}
	with $U\in\mathbb{R}^{l\times\bar{l}}$ being an isometry (i.e., $U'U=I$) and $\epsilon$ a small positive number, both to be determined.
    Design the parameter matrices of the decoder \eqref{eq:decoder} as
    \begin{equation}\label{eq:design_decoder}
        \begin{aligned}
            \left[\begin{array}{ccc}
		A_{\de} & B\\
		C & D 
	\end{array}\right]=\left[\begin{array}{ccc}
		A & \bar{L}\\
		I & 0 
	\end{array}\right].
        \end{aligned}
    \end{equation}
    Under the design of \eqref{eq:design_encoder} and \eqref{eq:design_decoder}, it can be verified that $\hat{x}(k)=x_e(k)+\bar{L}r(k-1)$,
    which implies  
    \begin{equation*}
        \begin{aligned}
            s(k)&=\bar{M}x_e(k)-\bar{M}x(k)+\bar{M}\bar{L}r(k-1)\\
            &=\bar{M}\hat{x}(k)-\bar{M}x(k)=-\bar{M}\tilde{x}(k).
        \end{aligned}
    \end{equation*}
    Then, the estimation error dynamics is given by
    \begin{align}
        \tilde{x}(k+1)&=x(k+1)-\hat{x}(k+1)\notag\\
        &=Ax(k)+v(k)-A\hat{x}(k)-\bar{L}(s(k)+w(k))\notag\\
        &=(A+\bar{L}\bar{M})\tilde{x}(k)-\bar{L}w(k)+v(k).\label{eq:err_system}
    \end{align}
    Since $A+\bar{L}\bar{M}=A+\tilde{L}\tilde{M}$ is Schur stable, $\tilde{x}(k)$ converges to a stationary process with a finite covariance matrix, thus achieving bounded EEC. 
    Accordingly, the channel input $s(k)=-\bar{M}\tilde{x}(k)$ also converges to a stationary process with a finite covariance matrix. 
    We next show that the subchannel power constraints \eqref{eq:subchannelC} are satisfied by appropriately choosing $U$ and $\epsilon$.
    With the encoder design \eqref{eq:design_encoder}, direct calculation gives
	\begin{align*}
		&\tf{T}_v(z)=-\bar{M}(zI-A-\bar{L}\bar{M})^{-1},\\
		&\tf{T}_w(z)=\bar{M}(zI-A-\bar{L}\bar{M})^{-1}\bar{L}.
	\end{align*}
	Then by \eqref{eq:thm_encoderdecoder}, we obtain
	\begin{align}
		&\frac{1}{2\pi}\int_{-\pi}^{\pi}\tf{T}_w(e^{j\omega})\Sigma\tf{T}_w(e^{j\omega})^*\,\mathrm{d}\omega\notag\\
		&=\Sigma^{\frac{1}{2}}U\tilde{M}\Big(\frac{1}{2\pi}\int_{-\pi}^{\pi}(e^{j\omega}I-A-\tilde{L}\tilde{M})^{-1}\tilde{L}\tilde{L}'\notag\\
		&\qquad\qquad\qquad\times(e^{-j\omega}I-A'-\tilde{M}'\tilde{L}')^{-1}\,\mathrm{d}\omega\Big) \tilde{M}'U'\Sigma^{\frac{1}{2}}\notag\\
        &=\Sigma^{\frac{1}{2}}U\tilde{M}X\tilde{M}'U'\Sigma^{\frac{1}{2}}\notag\\
		&=\Sigma^{\frac{1}{2}}U(\Gamma-I) U'\Sigma^{\frac{1}{2}},	\label{eq:thmTw}
	\end{align}
	where the second equality follows from that $\tilde{L}$ minimizes the $\mathcal{H}_2$ complementary sensitivity for $\left[\begin{array}{c}
    A  \\ \hline \\[-1em] 
     \tilde{M}
\end{array}\right]$ \cite{zhou1996robust},
        and the last equality is due to \eqref{eq:thmMu}.
	Denote the same-order or higher-order infinitesimals of $\epsilon$ by $O(\epsilon)$. Since $\bar{M}=O(\epsilon)$, we have
	\begin{equation}\label{eq:thmTv}
		\frac{1}{2\pi}\int_{-\pi}^{\pi}\tf{T}_v(e^{j\omega})V\tf{T}_v(e^{j\omega})^*\,\mathrm{d}\omega=O(\epsilon^2).
	\end{equation}
	Combining this with \eqref{eq:thmTw}, we obtain
	\begin{equation}\label{eq:thmSii}
		\lim_{k\rightarrow\infty}\mathbb{E}\{s_i(k)^2\}=\sigma_i^2\left\{U(\Gamma-I)U'\right\}_{ii}+O(\epsilon^2).
	\end{equation} 
    From \eqref{eq:cor_opt_P_sigma} and Lemma~\ref{lma:maj1}, it follows that there exists a $\eta=\begin{bmatrix}
		\eta_1&\eta_2&\cdots&\eta_{l}
	\end{bmatrix}'$ such that
	\begin{equation}\label{eq:thmPgesigma}
		\frac{\p_i}{\sigma_i^2}+1>\eta_i,\ i=1,\dots,l,
	\end{equation}
	and $\eta\preccurlyeq\begin{bmatrix}
		\gamma_1&\!\cdots\!&\gamma_{\bar{l}} &1 &\!\cdots\!&1
	\end{bmatrix}'$. 
    This means $\eta-\boldsymbol{1}_l\preccurlyeq\begin{bmatrix}
		\gamma_1-1&\!\cdots\!&\gamma_{\bar{l}}-1 &0 &\!\cdots\!&0
	\end{bmatrix}'$.
    By Lemma~\ref{lma:symmetric_matrix}, an isometry $U$ can be constructed such that
	\begin{equation}\label{eq:thmUgammaU}
		\{U(\Gamma-I)U'\}_{ii}=\eta_i-1,\ i=1,\dots,l.
	\end{equation} 
	When $\epsilon$ is selected to be sufficiently small, integrating \eqref{eq:thmSii}, \eqref{eq:thmPgesigma}, and \eqref{eq:thmUgammaU} yields that the subchannel power constraints \eqref{eq:subchannelC} are satisfied.

    Now consider the case where $A$ is a cyclic matrix with repeated eigenvalues.
     A real matrix $Y$ can still be constructed from \eqref{eq:cor_opt_gamma}. However, if this $Y$ is noncyclic, the previous proof no longer applies since $A$ and $Y$ have distinct Jordan canonical forms.
    To resolve this, we make the following modification.    
    In view of \eqref{eq:cor_opt_P_sigma}, we can choose a constant $0<\delta<1$ sufficiently close to $1$ such that
    \begin{equation}\label{eq:condition_delta2}
    \begin{aligned}
        &\begin{bmatrix}
	\frac{\p_1}{\sigma_1^2}+1&\frac{\p_2}{\sigma_2^2}+1&\cdots&\frac{\p_{l}}{\sigma_{l}^2}+1
					\end{bmatrix}' \prec^{\mathrm{w}}\\
        &\begin{bmatrix}
		\frac{\gamma_1}{\rule{0pt}{1.5ex}\delta^{(n-\bar{l})/\bar{l}}}&\frac{\gamma_2}{\rule{0pt}{1.5ex}\delta^{(n-\bar{l})/\bar{l}}}&\cdots&\frac{\gamma_{\bar{l}}}{\rule{0pt}{1.5ex}\delta^{(n-\bar{l})/\bar{l}}}&1 &1 &\cdots&1
	\end{bmatrix}'.
    \end{aligned}
    \end{equation}
    From \eqref{eq:cor_opt_gamma}, it follows that
    \begin{equation*}
    \begin{aligned}
        &\begin{bmatrix}
		|\lambda_1|^2&|\lambda_2|^2&\cdots&|\lambda_{n}|^2
		\end{bmatrix}' \underset{\mathrm{log}}\preccurlyeq\\
        &\begin{bmatrix}
		\frac{\gamma_1}{\rule{0pt}{1.5ex}\delta^{(n-\bar{l})/\bar{l}}}&\frac{\gamma_2}{\rule{0pt}{1.5ex}\delta^{(n-\bar{l})/\bar{l}}}&\cdots&\frac{\gamma_{\bar{l}}}{\rule{0pt}{1.5ex}\delta^{(n-\bar{l})/\bar{l}}}&\delta &\delta &\cdots&\delta
	\end{bmatrix}'.
    \end{aligned}
    \end{equation*}
    Then according to Lemma~\ref{lma:cyclic_matrix}, we can construct a real matrix $Y$ with eigenvalues $\lambda_1,\dots,\lambda_{n}$ and squared singular values $\frac{\gamma_1}{\rule{0pt}{1.5ex}\delta^{(n-\bar{l})/\bar{l}}},\dots,\frac{\gamma_{\bar{l}}}{\rule{0pt}{1.5ex}\delta^{(n-\bar{l})/\bar{l}}},\delta,\dots,\delta$.
    For any $\tau>0$, there exists a real perturbation matrix ${T}_\tau$ with entries of order $O(\tau)$ such that the perturbed matrix ${Y}_{\tau}=Y+{T}_\tau$ is cyclic and preserves the eigenvalues \cite{KlimenkoSergeichuk2014}. Then ${Y}_{\tau}$ and $A$ are similar, so there exists a real nonsingular matrix ${Z}_\tau$ such that ${Y}_\tau={Z}_\tau^{-1}A{Z}_\tau$.
    Since the singular values of a matrix depend continuously on its entries, a singular value decomposition of $Y_\tau$ yields
	\begin{equation*}\label{eq:svd2}
		{Y}_\tau={P}_\tau\begin{bmatrix}
			\Gamma_{\tau}&0\\
			0&\Delta_\tau
		\end{bmatrix}^\frac{1}{2}{Q}_\tau',
	\end{equation*}
	where ${P}_\tau,{Q}_\tau\in\mathbb{R}^{n\times n}$ are orthogonal matrices, and
    \begin{align}
            \Gamma_{\tau}&=\diag\left\{ \frac{\gamma_1}{\delta^{(n-\bar{l})/\bar{l}}},\dots,\frac{\gamma_{\bar{l}}}{\delta^{(n-\bar{l})/\bar{l}}}\right\}+O(\tau),\label{eq:Gamma_tau}\\
            \Delta_\tau&=\diag\{\delta,\dots,\delta\}+O(\tau)\label{eq:Delta_tau}
    \end{align}
     are diagonal matrices containing the squared singular values.
    Design
    \begin{equation}\label{eq:M_tau}
        {M}_\tau=\begin{bmatrix} 
	    (\Gamma_\tau-I)^{\frac{1}{2}}&0_{\bar{l}\times (n-\bar{l})}
	\end{bmatrix}{Q}'_\tau{Z}^{-1}_\tau.
    \end{equation}
         Using a derivation similar to \eqref{eq:derivation_DARE}, we obtain
        \begin{equation}\label{eq:MADR}
		{X}_\tau=A(I+{X}_\tau{M}_\tau'{M}_\tau)^{-1}{X}_\tau A'+{V}_\tau,
	\end{equation}
        where ${X}_\tau={Z}_\tau{Z}_\tau'>0$ and ${V}_\tau={Z}_\tau{P}_\tau\begin{bmatrix}
            0&0\\
            0&I-\Delta_\tau
        \end{bmatrix}{P}_\tau'{Z}_\tau'$.
    Let ${L}_\tau=-A(I+X_\tau{M}_\tau'{M}_\tau)^{-1}X_\tau{M}_\tau'$.
    Then, from \eqref{eq:MADR}, it follows that
    \begin{equation}\label{eq:LyapunovE}
        {X}_\tau=(A+{L}_\tau M_\tau)X_\tau (A+{L}_\tau M_\tau)'+L_\tau L_\tau'+V_\tau.
    \end{equation}
    By \eqref{eq:Delta_tau} and $\delta<1$, there exists a small $\tau_0 > 0$ such that for all $\tau<\tau_0$, we have $\Delta_{\tau} \leq I$, i.e., $V_\tau \geq 0$.
    Since also $X_\tau>0$, it follows from Lyapunov equation property \cite[Lemma~21.2]{zhou1996robust} that  $|\lambda_i(A+{L}_\tau M_\tau)|\leq1$. 
    Moreover, any eigenvalue $\alpha$ on the unit circle would require its corresponding left eigenvector $\xi$ to satisfy $\xi^*L_\tau=0$, implying $\xi^*A=\xi^*(A+L_\tau M_\tau) =\alpha \xi^*$. This contradicts the assumption $\vert{}\lambda_i(A)\vert{} > 1$. 
    Hence, $A+{L}_\tau M_\tau$ is Schur stable.
    Substituting $\tilde{M}$ and $\tilde{L}$ with ${M}_\tau$ and ${L}_\tau$ in \eqref{eq:thm_encoderdecoder}, and using the same encoder-decoder construction as \eqref{eq:design_encoder} and \eqref{eq:design_decoder}, we conclude that the estimation error dynamics is stable, and thus the EEC is bounded. Furthermore,
    \begin{align}
        &\frac{1}{2\pi}\int_{-\pi}^{\pi}\tf{T}_w(e^{j\omega})\Sigma\tf{T}_w(e^{j\omega})^*\,\mathrm{d}\omega=\Sigma^{\frac{1}{2}}U{M}_\tau S_\tau {M}_\tau'U'\Sigma^{\frac{1}{2}}\notag\\
        &\quad\leq\Sigma^{\frac{1}{2}}U{M}_\tau X_\tau {M}_\tau'U'\Sigma^{\frac{1}{2}}=\Sigma^{\frac{1}{2}}U(\Gamma_\tau-I) U'\Sigma^{\frac{1}{2}},\label{eq:cyclic_w_h2}
    \end{align}
    where $S_\tau$ is the solution to the Lyapunov equation
    \begin{equation*}
        S_\tau=(A+{L}_\tau M_\tau)S_\tau (A+{L}_\tau M_\tau)'+L_\tau L_\tau',
    \end{equation*}
    the inequality is due to \eqref{eq:LyapunovE} and $V_\tau\geq0$, and the last equality follows from \eqref{eq:M_tau}.
    Using \eqref{eq:condition_delta2}, \eqref{eq:Gamma_tau}, and \eqref{eq:cyclic_w_h2}, and repeating the previous analysis, an isometry $U$ can be constructed such that the subchannel power constraints are satisfied, provided that $\tau$ and $\epsilon$ are chosen sufficiently small. This completes the proof.
 \end{myproof}

The sufficiency proof provides a systematic procedure for designing an admissible encoder-decoder pair when the problem is solvable.
 The structure of the resulting pair is illustrated in Fig.~\ref{fig:encoder-decoder-design}, where the parameter matrices $\bar{M}$ and $\bar{L}$ are designed according to the following steps:
\begin{enumerate}
\item Based on the multiplicative majorization inequality \eqref{eq:cor_opt_gamma}, construct matrix $Y$ to relate the source antistable poles $\lambda_1,\dots,\lambda_{n_a}$ to the communication demands $\gamma_1,\dots,\gamma_{\bar{l}}$.
    \item Design $\tilde{M}$ as in \eqref{eq:thmMu}, and design $\tilde{L}$ accordingly from the DARE \eqref{eq:thmX_u} to minimize $\mathcal{H}_2$ complementary sensitivity for $\left[\begin{array}{c}
    A  \\ \hline \\[-1em] 
     \tilde{M}
\end{array}\right]$.
    \item Based on the additive majorization inequality \eqref{eq:cor_opt_P_sigma}, 
    construct an isometry matrix $U$ satisfying \eqref{eq:thmUgammaU} to even out the communication demands across the subchannels so that the subchannel power constraints are satisfied.
    \item Construct $\bar{M}$ and $\bar{L}$ as in \eqref{eq:thm_encoderdecoder}, where $\epsilon$ is a chosen small positive
number. 
\end{enumerate}

\begin{figure}[htbp]
	\centering
	\includegraphics[width=1\linewidth]{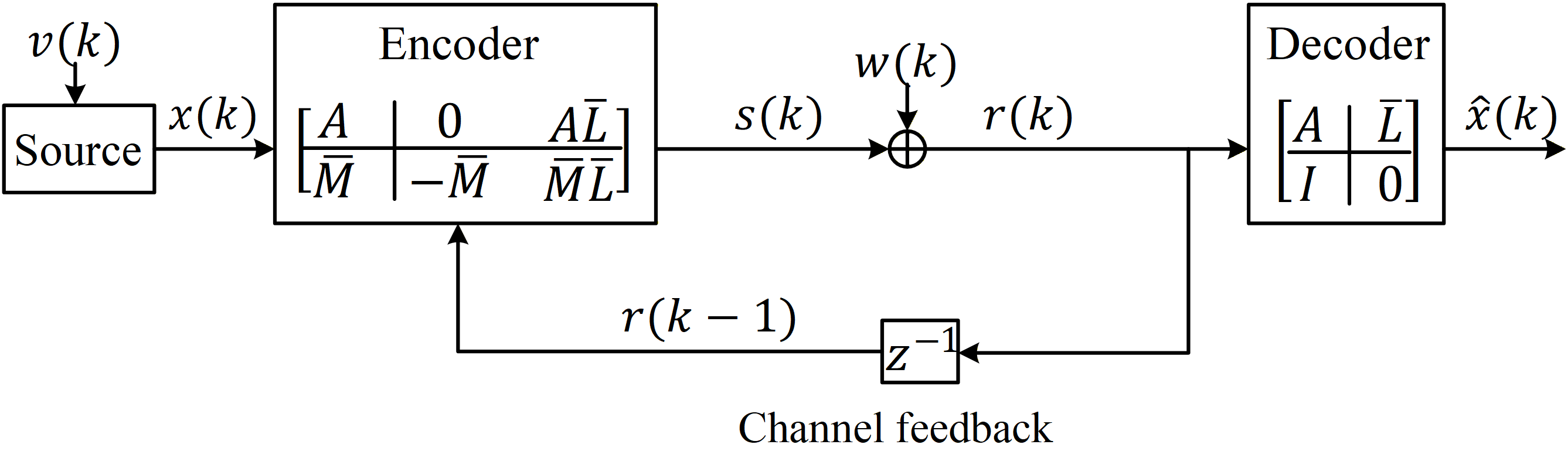}
	\caption{Structure of the proposed encoder-decoder pair.    }\label{fig:encoder-decoder-design}
\end{figure}

\section{Linear Coding With Total Power Constraint}\label{sec:total_power}

\subsection{Preliminary Analysis via Partial Order Programming under Majorization Order}
Now we turn to Problem 2, where the total power constraint \eqref{eq:totalC} is imposed. In this case, one has the degree of freedom to allocate the total power among the subchannels. Our goal is to determine the minimum total power under which an encoder-decoder pair can be designed to achieve bounded EEC subject to this constraint.

Denote  by $\p_i$ the power level allocated to the $i$th subchannel. Then, in view of Theorem~\ref{thm:subchannel}, Problem 2 is solvable if and only if there is an allocation of power level $\p_1,\dots,\p_l$ among the subchannels such that \eqref{eq:cor_opt_pro} is feasible. Thus, finding the minimum $\p$ required for the solvability of Problem 2 amounts to finding an optimal allocation $\p_1,\dots,\p_l$ that renders feasibility of \eqref{eq:cor_opt_pro} while minimizing $p=\sum_{i=1}^{l}\p_i$, namely,
\begin{equation}\label{eq:opt_total_p}
	\begin{aligned}
    &\inf_{\substack{\p_i,i=1,\dots,l\\
				\gamma_i\geq1,i=1,\dots,{\bar{l}}\\
                \text{s.t. } \eqref{eq:cor_opt_P_sigma} \text{ and } \eqref{eq:cor_opt_gamma} }} \sum_{i=1}^{l}\p_i.
	\end{aligned}
\end{equation} 

By Proposition \ref{Plambda}, the feasible set of problem (\ref{eq:opt_total_p}) is convex. As a result, (\ref{eq:opt_total_p}) is a convex optimization problem and can be solved numerically by convex optimization solvers. 

However, establishing convexity alone does not fully settle the problem. A more informative objective is to derive analytic solutions, or establish simple and interpretable rules to allocate power among the subchannels rather than relying solely on numerical computation.

To this end, 
we shall connect problem (\ref{eq:opt_total_p}) to a partial order progamming (POP) under majorization order.
Assume without loss of generality that the noise variances of the subchannels are ordered as
\begin{equation}\label{eq:noise_variance}
	{\sigma_1^2}\le{\sigma_2^2}\le\cdots\leq{\sigma_l^2},
\end{equation}
and the antistable eigenvalues of $A$ are ordered as
\begin{equation}\label{eq:eigenvalues_order}
    |\lambda_1|\geq|\lambda_2|\geq\cdots\geq|\lambda_{\anti}|.
\end{equation}
Denote $\log\sigma=\begin{bmatrix}
    \log\sigma_1&\log\sigma_2&\cdots&\log\sigma_{\bar{l}}
\end{bmatrix}^\prime$.

We formulate the following POP problem:
\begin{align}
&\underset{\ca\in\mathcal{F}}{\text{min}_\preccurlyeq}\;\;\; (\ca+\log\sigma),\label{eq:POP}\\
\text{with }&\mathcal{F}=\left\{   \ca\in\mathbb{R}_+^{\bar{l}}\cap\mathcal{D}^{\bar{l}}\:\middle|\:\log|\lambda|\: \preccurlyeq
            \begin{bmatrix}
                \ca\\
                \boldsymbol{0}
            \end{bmatrix}
    \right\},\nonumber
\end{align}
where the objective function is ordered in majorization order, and the constraint is characterized via a majorization inequality.

The existence and uniqueness of optimal solution in a POP problem is in general highly nontrivial. Interestingly, 
the POP problem (\ref{eq:POP}) does have a unique optimal solution.

\begin{proposition}\label{POPoptimal}
    There exists a unique optimal solution to the POP problem \eqref{eq:POP}.
\end{proposition}
\begin{proof}
Let $\mathcal{I}=\{1,2,\dots,\bar{l}\}$. For any nonempty $\mathcal{S}\subseteq\mathcal{I}$, denote by $\kappa(\mathcal{S})$ the smallest number in $\mathcal{S}$. 
Then consider the submodular function 
$f(\mathcal{S})\!=\!\TE(A)\!-\!\sum_{i=1}^{\kappa(\mathcal{S})-1}\log|\lambda_i|$ with $f(\emptyset)\!=\!0$. 
Its associated base polyhedron is given by
    \begin{equation*}
            \mathcal{B}_f\!=\!\left\{\!  
        \ca\!\in\!\mathbb{R}_+^{\bar{l}}\middle|\begin{aligned} 
            &\sum\nolimits_{i=1}^k\!c_i\!\geq\! \sum\nolimits_{i=1}^k\log\!|\lambda_i|,\,k=1,\dots,\bar{l}-1,\!\\
            &\sum\nolimits_{i=1}^{\bar{l}}\!c_i\!=\! \TE(A)
        \end{aligned}
    \right\}\!.
    \end{equation*}
    By Lemma~\ref{lma:base_poly}, there is a $c^\star$ that solves
    \begin{equation}\label{eq:POP_S}
\underset{c\in\mathcal{B}_f}{\text{min}_\preccurlyeq} \quad (c+\log\sigma).
	\end{equation}
    Note that $c^\downarrow\in\mathcal{B}_f$ for any $c\in\mathcal{B}_f$. From \eqref{eq:noise_variance} and \cite[6.A.2]{Marshall2011Majorization}, it follows that $c^{\star\downarrow}+\log\sigma\preccurlyeq c^\star+\log\sigma$, implying that $c^{\star\downarrow}$ is also optimal for \eqref{eq:POP_S}. 
    Since the feasible set $\mathcal{F}$ of the POP \eqref{eq:POP} is precisely $\mathcal{B}_f\cap \mathcal{D}^{\bar{l}}$, 
    $c^{\star\downarrow}$ is feasible and optimal for \eqref{eq:POP}.

    To prove the uniqueness of optimal solution of \eqref{eq:POP}, suppose for contradiction that there exists another optimal solution $c^\circ$. 
    Then $c^\circ+\log\sigma$ is a permutation of $c^{\star\downarrow}+\log\sigma$ but $c^\circ+\log\sigma\neq c^{\star\downarrow}+\log\sigma$.
    Let $\theta\in(0,1)$. By convexity of $\mathcal{F}$, $\theta c^\circ+(1-\theta) c^{\star\downarrow}\in\mathcal{F}$. 
    From \cite[2.B.3]{Marshall2011Majorization}, it follows that
        \begin{equation*}
            \begin{aligned}
                \theta c^\circ\!+\!(1\!-\!\theta) c^{\star\downarrow}+\log\sigma&\!=\!\theta (c^\circ+\log\!\sigma)\!+\!(1\!-\!\theta) (c^{\star\downarrow}+\log\sigma)\\
                &\!\preccurlyeq\! c^{\star\downarrow}+\log\sigma,
            \end{aligned}
        \end{equation*}
        implying that $\theta c^\circ+(1\!-\!\theta) c^{\star\downarrow}$ is also optimal for \eqref{eq:POP}. 
        However, this contradicts the fact that $\theta c^\circ+(1\!-\!\theta) c^{\star\downarrow}+\log\sigma$ is not a permutation of $c^{\star\downarrow}+\log\sigma$. The proof is completed.
        \end{proof}

The next proposition establishes the connection between the POP (\ref{eq:POP}) and the optimization problem  (\ref{eq:opt_total_p}).

\begin{proposition}\label{prop:total_power}
    Let $c^\star$ be the optimal solution to POP \eqref{eq:POP}. Then the optimal value of \eqref{eq:opt_total_p} is given by
    \begin{equation}\label{eq:p_star}   p^\star=\sum_{i=1}^{\bar{l}}\sigma_i^2(e^{2\ca_i^\star}-1),
    \end{equation}
    which can be approached arbitrarily closely as 
    \begin{align*}
    &p_i\rightarrow \sigma_i^2(e^{2\ca_i^\star}-1), i=1,\dots,\bar{l} \text{ and } \\
    &p_i\rightarrow0, i=\bar{l}+1,\dots,l.
    \end{align*}
\end{proposition}

\begin{proof}
    For any feasible solution $\p_1,\dots,\p_l$ of \eqref{eq:opt_total_p}, there exist $\gamma_i\geq 1$, $i=1,\dots,\bar{l}$ satisfying \eqref{eq:cor_opt_P_sigma} and \eqref{eq:cor_opt_gamma}. Without loss of generality, we assume that $\gamma_1\geq\cdots\geq\gamma_{\bar{l}}$. 
    From \eqref{eq:noise_variance}, \eqref{eq:cor_opt_P_sigma}, and Lemma~\ref{lma:schurconvex}, it follows that 
    	\begin{equation*}
		\sum_{i=1}^{l}\p_i=\sum_{i=1}^{l}\sigma_i^2\left(\frac{\p_i}{\sigma_i^2 }+1\right)-\sum_{i=1}^{l}\sigma_i^2>\sum_{i=1}^{\bar{l}}\sigma_i^2(\gamma_i-1),
	\end{equation*}
    where the lower bound is approached arbitrarily closely as 
    	\begin{equation}\label{eq:lower_bound_solution}
		\begin{aligned}
			\begin{cases}
				p_i\rightarrow \sigma_i^2(\gamma_i-1),&i=1,\dots,\bar{l}, \\
				p_i\rightarrow 0, &i=\bar{l}+1,\dots,l.\
			\end{cases}
		\end{aligned}
	\end{equation}      
    Thus, the optimal value of \eqref{eq:opt_total_p} equals the minimum of 
    \begin{equation*}
	\begin{aligned}
		&\min_{\substack{\gamma_1 \geq \cdots \geq \gamma_{\bar{l}} \geq 1\\
        \text{s.t. }\eqref{eq:cor_opt_gamma}}} \sum_{i=1}^{\bar{l}}\sigma_i^2(\gamma_i-1).
	\end{aligned}
    \end{equation*}
    By applying the change of variables $\gamma_i=e^{2c_i}$ for $i=1,\dots,\bar{l}$, this problem is transformed into
    \begin{equation}\label{eq:opt_scalar_power}
\underset{\ca\in\mathcal{F}}{\text{min}} \ \  \sum_{i=1}^{\bar{l}}\sigma_i^2(e^{2c_i}-1) \text{, i.e., } \underset{\ca\in\mathcal{F}}{\text{min}} \ \  \sum_{i=1}^{\bar{l}}(e^{2(c_i+\log\sigma_i)}-\sigma_i^2).
\end{equation}
    Since $(x_1,\dots,x_{\bar{l}})\mapsto\sum_{i=1}^{\bar{l}}e^{2x_i}$ is Schur convex \cite{Marshall2011Majorization}, $c^\star$ is optimal for \eqref{eq:opt_scalar_power}. By \eqref{eq:lower_bound_solution}, the proof is completed.
\end{proof}

 It is now clear that Problem 2 is solvable if and only if 
the total power constraint $p\!>\!p^\star$.
Moreover, for a given $p\!>\!p^\star$, the power constraint can be allocated across the subchannels such that
\[
\p_i>\sigma_i^2(e^{2\ca_i^\star}-1) \text{ for }i=1,\dots,\bar{l},
\]
and $p_i>\epsilon$ for the remaining $l-\bar{l}$ subchannels, where $\epsilon>0$ is a small number. 
Under such an allocation, condition \eqref{eq:cor_opt_pro} is satisfied by letting $\gamma_i = e^{2\ca_i^\star}$. Consequently, a feasible encoder-decoder pair can  be constructed following the procedure in the sufficiency proof of Theorem~\ref{thm:subchannel}.

\begin{remark}\label{lna}
The above analysis reveals an interesting observation: When $l> n_a$, i.e., the number of available subchannels exceeds that of the antistable poles, the minimum total transmission power can be approached arbitrarily closely by using the $n_a$ subchannels with the lowest noise variances; the use of the remaining $l-n_a$ subchannels provide no further power reduction.
\end{remark}

With Problem 2 translated into solving POP \eqref{eq:POP}, we shall first consider the case where all subchannel noise variances are identical, i.e., $\sigma_1^2 = \sigma_2^2 = \cdots = \sigma_l^2$. For this case, we derive an analytic solution and provide a water-filling interpretation. We then extend the investigation to the general case with unequal noise variances, for which a sequential water-filling algorithm is developed.

\subsection{The Case of Equal Subchannel Noise Variance: Analytic Solution and Water-Filling Interpretation}

In light of Remark \ref{lna}, we assume without loss of generality that $l\leq n_a$, i.e., $\bar{l}=l$, in the sequel.

Under equal noise variances,
the optimal solution to the POP (\ref{eq:POP}) is identical to that of
\[
\underset{\ca\in\mathcal{F}}{\text{min}_\preccurlyeq}\;\;\; \ca,
\]
namely, the least element of the poset $(\mathcal{F},\preccurlyeq)$. 
The existence of this least element follows from Proposition
\ref{POPoptimal}. 
The following theorem gives an analytic expression for this optimal solution.

\begin{theorem}\label{thm:join}
When the subchannel noise variances are equal, the optimal solution to POP \eqref{eq:POP} is given by
\begin{equation}\label{eq:least_element1}
\textstyle
    \!\!\ca^\star=\begin{bmatrix}
        \log\!|\lambda_1|\!&\!\cdots\!&\!\log\!|\lambda_{l-1}|\!&\!0
    \end{bmatrix}'\lor     \frac{\TE(A)}{l}\!\begin{bmatrix}
           1\!&\!1&\!\cdots\!&\!1\end{bmatrix}'\!.
\end{equation}
where $\lor$ denotes the join of two vectors under $\preccurlyeq_{\mathrm{w}}$.
\end{theorem}

\begin{proof}
    It suffices to show $\ca^\star$ in \eqref{eq:least_element1} is the least element of $(\mathcal{F},\preccurlyeq)$.
    Denote
    \begin{equation*}
        \begin{aligned}
            u\!=\!\begin{bmatrix}
        \log\!|\lambda_1|\!&\!\cdots\!&\!\log\!|\lambda_{l-1}|\!&\!0
    \end{bmatrix}'\!\text{ and }
    v\!=\!{\textstyle \frac{\TE(A)}{l}}\!\begin{bmatrix}
           1\!&\!1&\!\cdots\!&\!1\end{bmatrix}'.\!
        \end{aligned}
    \end{equation*}
    Let $\mathcal{U}$ be the set of upper bounds of $\{u,v\}$ in $(\mathcal{D}^l,\preccurlyeq_{\mathrm{w}})$.
    For any $\ca\in\mathcal{F}$, we have $u\preccurlyeq_{\mathrm{w}}\ca$ and $v\preccurlyeq_{\mathrm{w}}\ca$,
    so $\ca\in\mathcal{U}$. Since $\ca^\star$ is the least element of $\mathcal{U}$, it follows that 
    \begin{equation}\label{eq:a<=c}
         \ca^\star\preccurlyeq_{\mathrm{w}}\ca,\ \forall \ca\in\mathcal{F},
    \end{equation}
    and hence $\sum_{i=1}^{l}\ca^\star_i\leq\sum_{i=1}^{l}\ca_i=\TE(A)$. On the other hand, since $\ca^\star\in\mathcal{U}$, we have 
$u\preccurlyeq_{\mathrm{w}}\ca^\star$
    and $\TE(A)\leq \sum_{i=1}^{l}\!\ca^\star_i$. 
    Thus $\sum_{i=1}^{l}\!\ca^\star_i=\TE(A)$; together with \eqref{eq:a<=c}, this means $\ca^\star\preccurlyeq\ca$ for all  $\ca\in\mathcal{F}$. Moreover, from \eqref{eq:eigenvalues_order} and $u\preccurlyeq_{\mathrm{w}}\ca^\star$, we conclude $\ca^\star\in\mathcal{F}$. This completes the proof.
\end{proof}

The method for computing the join $\lor$ under $\preccurlyeq_{\mathrm{w}}$ is given in Appendix~\ref{appendix:join}. 

For the computation of $c^\star$ as in (\ref{eq:least_element1}), it admits an intuitive water-filling interpretation, as illustrated in Fig.~\ref{fig:water-filling}.
In particular, 
the vector $\begin{bmatrix}
\log|\lambda_1| & \cdots & \log|\lambda_{l-1}| & 0
\end{bmatrix}^{\prime}$
serves as a base profile, whose components represent the baseline workloads assigned to the respective subchannels.
The remaining workload $\sum_{i=l}^{n_a}\log\vert{}\lambda_i\vert{}$ is then allocated across the subchannels so that the resulting overall workload profile is as balanced as possible. This can be visualized as pouring the remaining workload like water onto the base profile, producing a water level $h$. All components below $h$ are raised to $h$, while those already above $h$ remain unchanged. The resulting profile gives the optimal allocation $c^\star$.

\begin{figure}[htbp]
	\centering
	\includegraphics[width=0.95\linewidth]{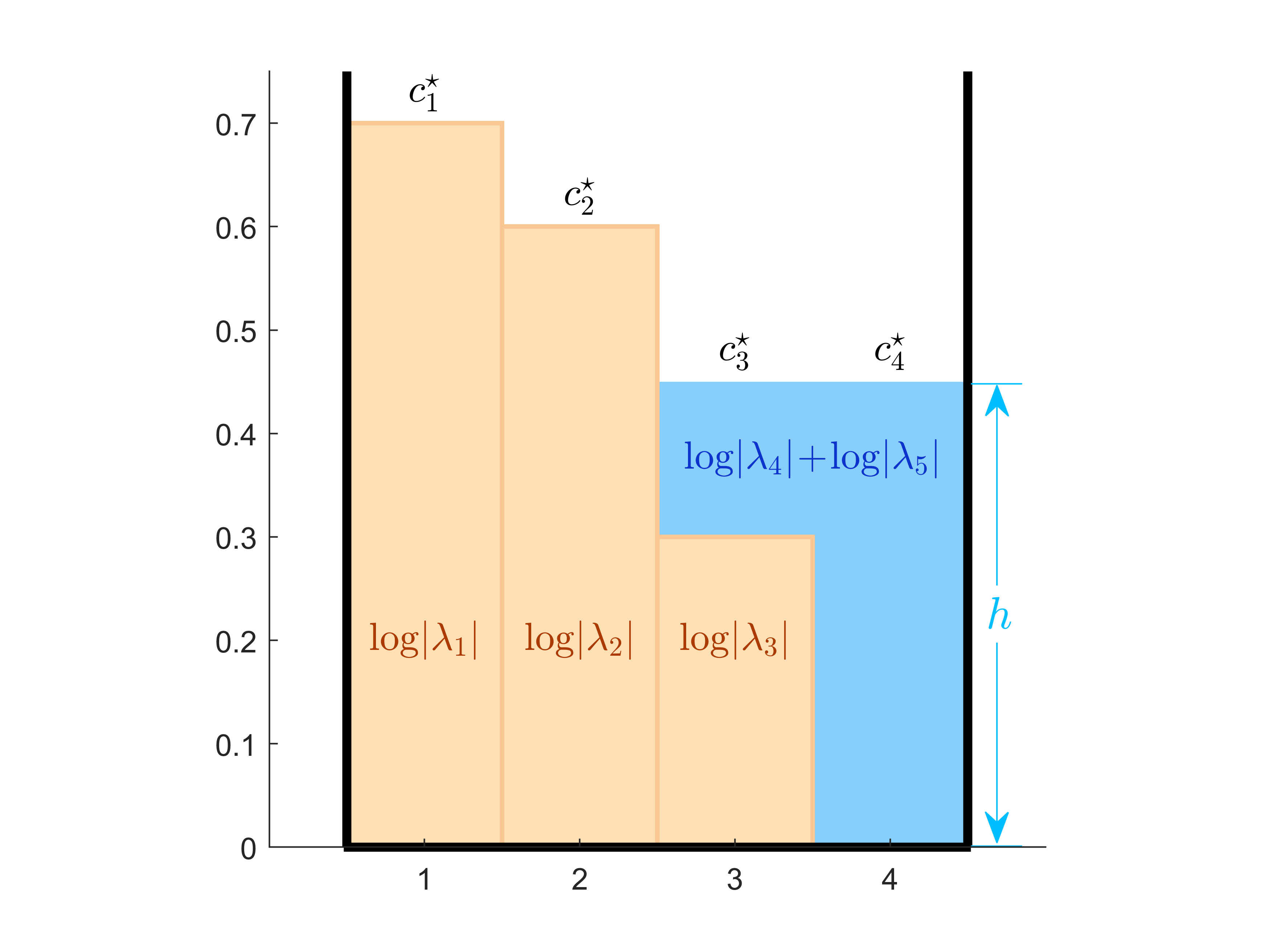}
    \caption{Water-filling interpretation for computation of 
    $\ca^\star$. Consider a source with antistable poles $\lambda=\begin{bmatrix}
 	    e^{0.7}\!\!&\!\!e^{0.6}\!\!&\!\!e^{0.3}\!\!&\!\!e^{0.3}\!\!&\!\!e^{0.3}
 	\end{bmatrix}^\prime$ and a
    parallel Gaussian channel comprising four subchannels with equal noise variances. As illustrated, $0.6$ units of water (i.e., $\log\!|\lambda_4|+\log\!|\lambda_5|$) are poured into a vessel with a base profile given by $\begin{bmatrix}
        0.7\!\!&\!\!0.6\!\!&\!\!0.3\!\!&\!\!0
    \end{bmatrix}^\prime$, yielding the optimal solution  $\ca^\star\!\!=\!\!\begin{bmatrix}
         0.7\!\!&\!\!0.6\!\!&\!\!0.45\!\!&\!\!0.45
     \end{bmatrix}^\prime$.
    }
	\label{fig:water-filling}
\end{figure}

In view of the water-filling procedure, all components of $c^\star$ are positive, indicating that the optimal solution makes use of all subchannels (under the premise that $l\!\leq\! n_a)$. Any strategy using only part of the subchannels is suboptimal, requiring a higher total power. This gives an important message:

\emph{Using more subchannels collectively for transmission of a discrete-time LTI source reduces total power requirement.} 

This contrasts the corresponding understanding in transmitting a continuous-time LTI cyclic source \cite{jin2024remote}, where a single subchannel suffices to achieve the minimum power, and using additional subchannels yields no further savings.

Before proceeding, we give an example comparing the exact minimum total power obtained with the upper bound in \cite{han2024coded}.
\begin{exam}\label{exam:comparison}
Consider an unstable source with dynamic matrix $A=\begin{bmatrix}\lambda&1&0\\0&\lambda&1\\0&0&\lambda\end{bmatrix}$, where $|\lambda|>1$. Consider the scenario of two subchannels with identical noise variances $\sigma_1^2=\sigma_2^2=\sigma^2$.
By Theorem~\ref{thm:join}, the optimal allocation $\ca^\star$ is given by
\begin{equation*}
    \ca^\star=\begin{bmatrix}
        \log|\lambda|\\
        0
    \end{bmatrix}\lor    \begin{bmatrix}
            \frac{3}{2}\log|\lambda|\\
            \frac{3}{2}\log|\lambda|\end{bmatrix}=\begin{bmatrix}
            \frac{3}{2}\log|\lambda|\\
            \frac{3}{2}\log|\lambda|\end{bmatrix},
\end{equation*}
i.e., allocating the total stabilizing workload $3\log|\lambda|$ evenly over the two subchannels. This yields the minimum total power $p^\star=\sigma^2(2|\lambda|^3-2)$. 
In contrast, the sufficient condition in \cite{han2024coded} requires the existence of a partition of the source's antistable poles such that the capacity of each subchannel exceeds the sum of log-magnitudes of its assigned poles.
For this scenario, there are two partitioning strategies: assigning all three poles to a single subchannel, or assigning two poles to one subchannel and the remaining one pole to the other.
Of the two partitions, the latter is more power-efficient,
leading to a total power of $\sigma^2(|\lambda|^4+|\lambda|^2-2)$, which is strictly greater than the minimum total power $p^\star$. Fig.~\ref{fig:comparison} illustrates the gap from such upper bound obtained in \cite{han2024coded} to $p^\star$. One can see that the gap increases as $|\lambda|$ increases. Also, the relative gap to $p^\star$ can be calculated as
\begin{align*}
    \frac{\sigma^2(|\lambda|^4+|\lambda|^2-2)-p^\star}{p^\star}=\frac{|\lambda|^2(|\lambda|-1)}{2(|\lambda|^2+|\lambda|+1)},
\end{align*}
which also increases as $\lambda$ is further away from the unit circle.
\end{exam}

\begin{figure}[htbp]
	\centering
	\includegraphics[width=0.8\linewidth]{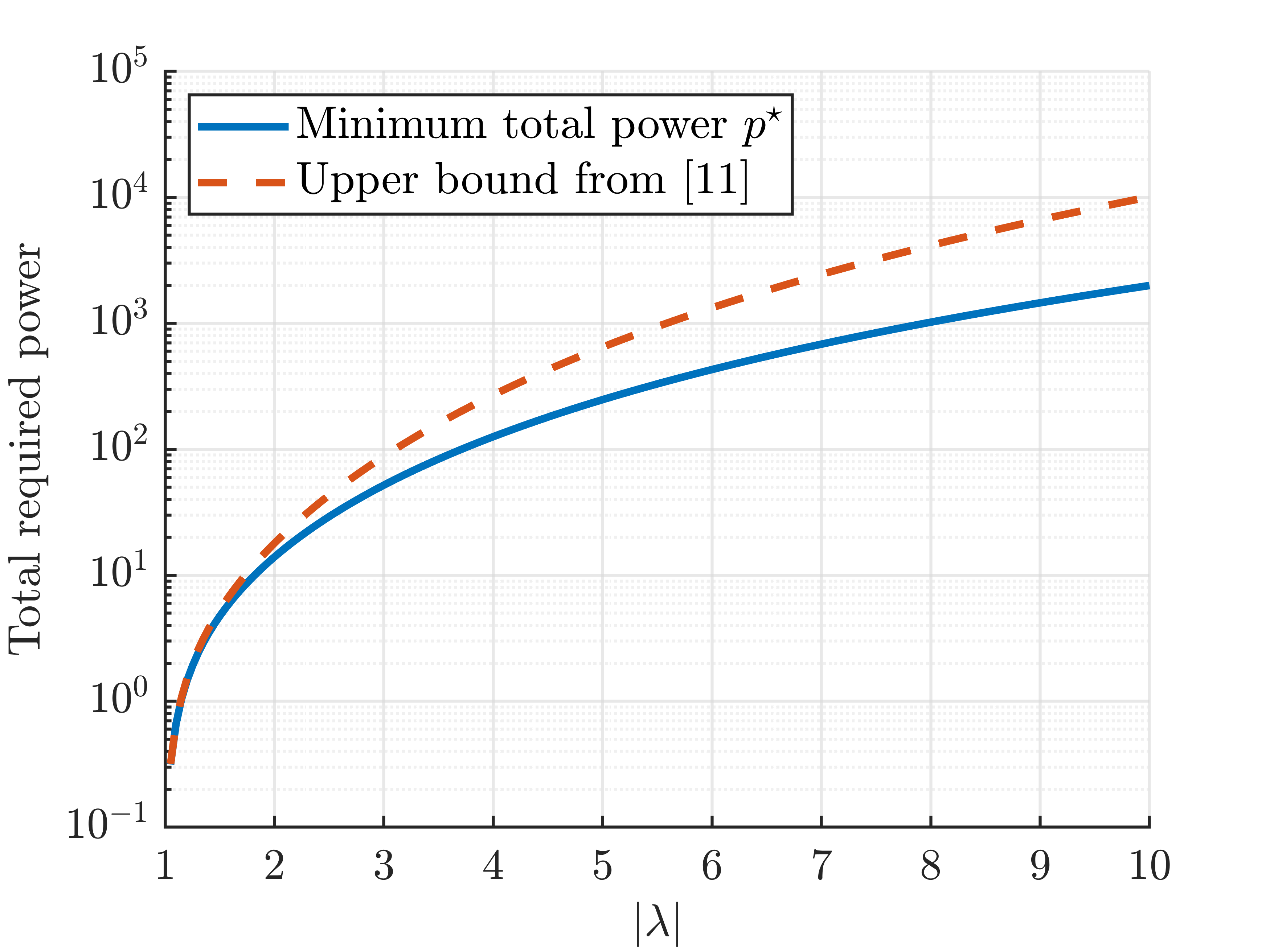}
    \caption{Comparison between the minimum total power $p^\star$ and the upper bound in \cite{han2024coded} for Example~\ref{exam:comparison} with $\sigma^2=1$.
    }
	\label{fig:comparison}
\end{figure}

\subsection{The General Case: Sequential Water Filling}

Now we extend the discussion to the general case where the subchannel noise variances may be unequal. 
Unlike the equal-noise case, the objective function in POP \eqref{eq:POP} can no longer be simplified to vector $c$ alone. Nevertheless, we shall show that the general case can still be addressed through water-filling, but performed in a sequential manner.

\begin{theorem}\label{thm:least}
For POP \eqref{eq:POP} in the general case, the optimal solution can be obtained via the sequential water-filling algorithm presented in Algorithm~\ref{alg:SWF}.
\end{theorem}

\begin{algorithm}[htbp]
\SetCommentSty{itshape} 
    \caption{Sequential Water-Filling Algorithm} \label{alg:SWF}
    \KwIn{$\eta=\begin{bmatrix}
    \log|\lambda_1|\!&\!\cdots\!&\!\log|\lambda_{l-1}|\!&\!\sum_{i=l}^{\anti}\log|\lambda_i|
    \end{bmatrix}^\prime$
    $\qquad\quad\log\sigma=\begin{bmatrix}
        \log\sigma_1\!&\!\log\sigma_2 \!&\!\cdots \!&\! \log\sigma_l
    \end{bmatrix}^\prime$}
    \KwOut{$c^\star$}
    Initialize: $y^{(0)}\gets\begin{bmatrix}\ \end{bmatrix}$\;
    \For{$k\gets 1$ \KwTo $l$}{
    \tcp{\small Perform $k$th water-filling iteration}
        Form the $k$th-stage base profile: $b^{(k)}\gets\begin{bmatrix}
            y^{(k-1)}\\
            \log\sigma_k
        \end{bmatrix}$\;\vspace{0.7ex}
        Find water level $h_k$ by $\sum_{i=1}^{k}\!\max\{0,h_k-b_i^{(k)}\}\!=\!\eta_k$;\!\!\!\!\\ \vspace{0.7ex}
        Update: $y^{(k)}\gets \begin{bmatrix}
            \max\{b_1^{(k)},h_k\}\\
            \vdots\\
            \max\{b_k^{(k)},h_k\}\end{bmatrix}$\;
        }
     \Return{$y^{(l)}-\log\sigma$}\;
\end{algorithm}

Before proving the above theorem, we first provide an intuitive interpretation on the sequential water-filling algorithm, as illustrated in Fig.~\ref{fig:calculation_swf}.
 The components of the vector $\eta$ can be viewed as stabilization workloads to be distributed across the subchannels. Specifically, the workloads corresponding to the largest $l-1$ antistable poles (in terms of magnitudes) come one by one, whereas the workloads corresponding to the remaining antistable poles come in aggregate. Then the POP \eqref{eq:POP} seeks to distribute these workloads over the unequal noise floors so that the resulting effective levels
\[
y_i=c_i+\log \sigma_i
\]
are as balanced as possible.
This is carried out in a sequential manner. At the $k$th step, the first $k-1$ subchannels have already formed an effective level profile $y^{(k-1)}$, 
and the $k$th subchannel is added with noise floor $\log\sigma_k$. 
The workload $\eta_k$ is then poured over a combined base profile
\[
b^{(k)}=
\begin{bmatrix}
y^{(k-1)}\\
\log\sigma_k
\end{bmatrix}.
\]
All components below the resulting water level are raised to this level, 
while those already above it remain unchanged. 
The updated effective level profile $y^{(k)}$ is then carried to the next step. 
After $l$ steps, the optimal allocation is obtained as
\[
c^\star = y^{(l)}-\log\sigma.
\]

\begin{remark} 
When specialized to the previous case of equal noise levels, the first $l-1$ water-filling steps become trivial, and the final step reduces to the procedure depicted in Fig.~\ref{fig:water-filling}.
\end{remark}

\begin{figure*}[htbp]
    \centering
    \includegraphics[width=0.95\linewidth]{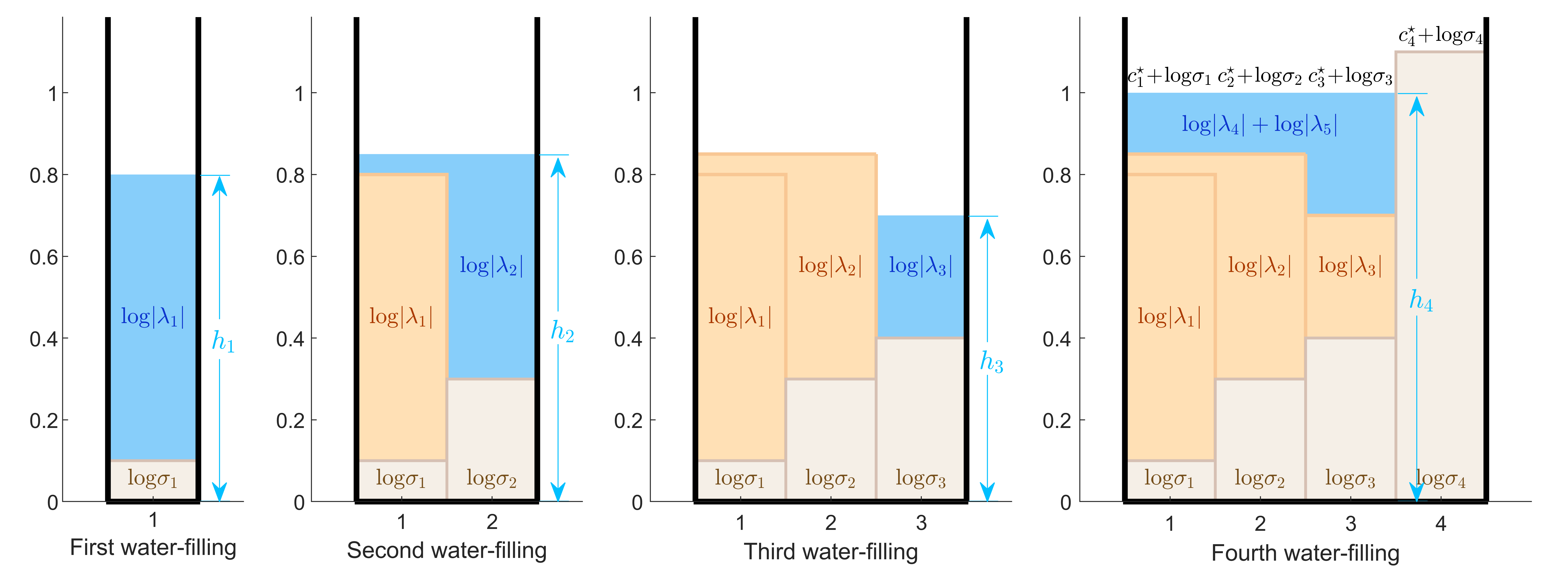}
    \caption{Computation of the optimal solution $\ca^\star$ via the sequential water-filling algorithm. Consider four subchannels with noise parameters $\log\sigma=\begin{bmatrix}
        0.1\!&\!0.3\!&\!0.4\!&\!1.1
    \end{bmatrix}^\prime$, and a source with five antistable poles characterized by $\log|\lambda|=\begin{bmatrix}
        0.7\!&\!0.6\!&\!0.3\!&\!0.3\!&\!0.3
    \end{bmatrix}^\prime$. 
    The figure depicts the four water-filling iterations performed sequentially, yielding $\ca^\star=\begin{bmatrix}
        0.9\!&\!0.7\!&\!0.6\!&\!0
    \end{bmatrix}^\prime$.}
    \label{fig:calculation_swf}
\end{figure*}

\begin{myproof}[Theorem~\ref{thm:least}]
    Let $d^{(k)}=y^{(k)}-\begin{bmatrix}
        \log\sigma_1\!&\!\cdots\!&\!\log\sigma_{k}
    \end{bmatrix}^\prime$ for $k=1,\dots,l$. 
    Then $c^\star=d^{(l)}$. We first prove $c^\star\in\mathbb{R}_+^l\cap\mathcal{D}^l$ via induction.
    Clearly, $d^{(1)}=\log|\lambda_1|\in\mathbb{R}_{+}$.
    Suppose that $d^{(k-1)}\in\mathbb{R}_+^{k-1}\cap\mathcal{D}^{k-1}$.
    According to the update rule,
    \begin{align*}
       y_i^{(k)}=\begin{cases}
            \max\{y^{(k-1)}_i,h_k\},&i=1,\dots,k-1,\\
            \max\{\log\sigma_k,h_k\}, &i=k.
        \end{cases}
    \end{align*}
    It then follows that
        \begin{align}\label{eq:d^k}
       \!\!\! d^{(k)}_i=\begin{cases}
            \max\{d^{(k-1)}_i,h_k-\log\sigma_i\},&i=1,\dots,k-1,\\
            \max\{0,h_k-\log\sigma_k\}, &i=k.
        \end{cases}
    \end{align}
    By $d^{(k-1)}\in\mathbb{R}_+^{k-1}$, we have $d^{(k)}\in\mathbb{R}_+^{k}$. By $d^{(k-1)}\in\mathcal{D}^{k-1}$ and \eqref{eq:noise_variance}, we obtain $d^{(k)}\in\mathcal{D}^{k}$, completing the induction. We next prove $c^\star\in\mathcal{F}$. 
    Recall that the water level $h_k$ satisfies
        \begin{equation*}
         \sum_{i=1}^{k-1}\max\{0,h_k-y_i^{(k-1)}\}+\max\{0,h_k-\log\sigma_k\}=\eta_k.
        \end{equation*}
    Using the identity $\max\{\alpha,\beta\} = \alpha + \max\{0, \beta-\alpha\}$, we have
    \begin{equation*}
        \begin{aligned}
        \sum_{i=1}^ky^{(k)}_i&=\sum_{i=1}^{k-1}\max\{y_i^{(k-1)},h_k\}+\max\{\log\sigma_k,h_k\}\\
        &=\sum_{i=1}^{k-1}y_i^{(k-1)}+\log\sigma_k+\eta_k.
        \end{aligned}
    \end{equation*}
    Applying this relation recursively yields that for $k=1,\dots,l$,  $\sum_{i=1}^k\!y^{(k)}_i\!=\!\sum_{i=1}^{k}\!(\log\sigma_i+\eta_i)$, implying $\sum_{i=1}^k\!d^{(k)}_i\!=\!\sum_{i=1}^{k}\!\eta_i$.
    Moreover, by \eqref{eq:d^k}, we obtain that for $k=1,\dots,l-1$,
      \begin{equation*}
        \sum_{i=1}^kd^{(l)}_i\geq\sum_{i=1}^kd^{(k)}_i=\sum_{i=1}^{k}\eta_i=\sum_{i=1}^{k}\log|\lambda_i|,
    \end{equation*}
    which together with \eqref{eq:eigenvalues_order} implies $\log|\lambda|\: \preccurlyeq
            \begin{bmatrix}
                c^\star\\
                \boldsymbol{0}
            \end{bmatrix}$, i.e., $\ca^\star\in\mathcal{F}$.

    We now establish the optimality of $c^\star$. Let $c\in\mathcal{F}$. Define 
    \begin{equation*}
        z^{(k)}=\begin{bmatrix}
        c_1+\log\sigma_1\!&\!\cdots\!&\!c_{k}+\log\sigma_{k}
    \end{bmatrix}^\prime,\ k=1,\dots,l.
    \end{equation*}
    Then $y^{(1)}=\log\!|\lambda_1|+\log\!\sigma_1\leq z^{(1)}$. Suppose for induction that $y^{(k-1)}\preccurlyeq_{\mathrm{w}}z^{(k-1)}$. 
    Then $\begin{bmatrix}
        y^{(k-1)}\\ \log\sigma_k
    \end{bmatrix}\preccurlyeq_{\mathrm{w}}\begin{bmatrix}
        z^{(k-1)}\\ c_k+\log\sigma_k
    \end{bmatrix}=z^{(k)}$. Also since
    \begin{equation*}
    \begin{aligned}
        \sum_{i=1}^{k-1}y^{(k-1)}_i\!+\!\log\sigma_k\!+\!\eta_k=\sum_{i=1}^{k}(\log\sigma_i\!+\!\eta_i)\leq\sum_{i=1}^{k}(\log\sigma_i\!+\!c_i),   
    \end{aligned}
    \end{equation*}
     applying Lemma~\ref{lma:w_f} (Appendix~\ref{appdix:key_lemma}) yields $y^{(k)}\preccurlyeq_{\mathrm{w}}z^{(k)}$. Thus, 
     \begin{equation*}
c^\star+\log\sigma=y^{(l)}\preccurlyeq_{\mathrm{w}}z^{(l)}= c+\log\sigma \text{ for any } c\in\mathcal{F},
     \end{equation*}
     which together with $\sum_{i=1}^lc^\star_i=\sum_{i=1}^lc_i=\TE(A)$ completes the proof.
\end{myproof}

\section{Numerical Examples}\label{sec:numerical_example}

Consider an unstable source described by \eqref{eq:source}, where
\begin{equation*}
A=\begin{bmatrix}
  e^{1.2} &0 &-1.2 &2.0\\
    0 &e &-1.1& 0.5\\
    0 &0 &e^{0.55} &1.3\\
    0 &0& 0& e^{0.5}  
\end{bmatrix},
 \end{equation*}
and the process noise has covariance matrix
\begin{equation*}
        V=\begin{bmatrix}
        0.91 &   0.16 &  -0.08  &  0.08\\
    0.16  &  0.94  &  0.12 &   0.22\\
   -0.08  &  0.12  &  0.75   &    0\\
    0.08  &  0.22  &     0  &  0.96
\end{bmatrix}.
\end{equation*}
The parallel AWGN channel comprises three subchannels with noise variances $\sigma_1^2=e^{0.1}$, $\sigma_2^2=e$, and $\sigma_3^2=e^{2}$.

\subsection{Individual Subchannel Power Constraints}

We first consider the coding problem under individual power constraints \eqref{eq:subchannelC} with 
    $\p_1=20$, $\p_2=20$, and $\p_3=30$.
The condition \eqref{eq:cor_opt_pro} is then satisfied by letting 
$\gamma_1\!=\!e^{3},\gamma_2\!=\!e^{2},\gamma_3\!=\!e^{1.5}$.
Hence Problem~1 is solvable.
An encoder-decoder design following the procedure in the sufficiency proof of Theorem~\ref{thm:subchannel} is provided below. Specifically, construct a matrix
\begin{equation*}
    Y=\begin{bmatrix}
        3.3201  &  0.9834 &   0.7517 &   1.2064\\
         0  &  2.7183  &  0.8176  &  1.3121\\
         0    &     0  &  1.7333   & 1.8557\\
         0   &      0   &      0  &  1.6487\\
    \end{bmatrix}
\end{equation*}
that is similar to matrix $A$, with singular values given by $\gamma_1^{\frac{1}{2}}$, $\gamma_2^{\frac{1}{2}}$, $\gamma_3^{\frac{1}{2}}$, and $1$.
Then compute $\tilde{M}$ and $\tilde{L}$ using
\eqref{eq:thmMu} and \eqref{eq:thmLu}, respectively.
Further, construct an isometry matrix
\begin{equation*}
    U=\begin{bmatrix}
  0.9504  &  0.3112      &   0\\
   -0.2950  &  0.9010   & 0.3180\\
    0.0989  & -0.3022    &0.9481\\
    \end{bmatrix}
\end{equation*}
according to \eqref{eq:thmUgammaU}. 
 With $\epsilon = 0.01$, \eqref{eq:thm_encoderdecoder} gives
 \begin{equation*}
 \begin{aligned}
          \bar{M}&\!=\!\begin{bmatrix}
         -3.0337  &  1.6238 &   0.4779 &  -3.3346\\
   -2.0356  &  3.2859 &  -3.6627  &-10.1768\\
    1.1257 &   0.4795 &  -7.3331   &-4.9006\\
     \end{bmatrix}\!\times\! 10^{-2}\text{ and}\\
     \bar{L}&\!=\!\begin{bmatrix}
         2.2192  &  2.0603   & 0.2795&    0.1358\\
   -1.1702  & -1.9319  & -0.1931 &  -0.0938\\
    0.7557  &  1.4269 &   0.3184 &   0.1546
     \end{bmatrix}^\prime\!\times\! 10^2.
 \end{aligned}
 \end{equation*}
The encoder and decoder are then implemented as in
Fig.~\ref{fig:encoder-decoder-design}.
 Under this design, the poles of the estimation error system \eqref{eq:err_system} are located at $e^{-1.2}$, $e^{-1}$, $e^{-0.55}$, and $e^{-0.5}$, which are the reciprocals of the eigenvalues of $A$, thereby confirming the boundedness of the EEC. 
To evaluate transmission power, we perform 50,000 Monte Carlo simulations over a time horizon of $T = 100$. The simulated average transmission powers of the subchannels are
 \begin{equation*}
 \begin{aligned}
     \frac{1}{T}\sum\nolimits_{k=0}^{T-1}\mathbb{E}\{s_1(k)^2\}&=19.7543 <\p_1,\\
    \frac{1}{T}\sum\nolimits_{k=0}^{T-1}\mathbb{E}\{s_2(k)^2\}&=19.5328 <\p_2,\\
    \frac{1}{T}\sum\nolimits_{k=0}^{T-1}\mathbb{E}\{s_3(k)^2\}&=28.8552<\p_3,
 \end{aligned}
 \end{equation*}
verifying that the subchannel power constraints are satisfied.

\subsection{Total Channel Power Constraint}

We now consider the coding problem under a total channel power constraint. 
By Proposition~\ref{prop:total_power}, the minimum total power required for Problem~2 is obtained by solving the POP~\eqref{eq:POP}. 
By Theorem~\ref{thm:least}, the sequential water-filling algorithm produces the optimal solution $\ca^\star\!=\!\begin{bmatrix}
    1.55\!&\!   1.10\!&\!    0.60
\end{bmatrix}^\prime$ to the POP problem, yielding the minimum required total power $\p^\star = 62.3851$.

Now we set the power limit in \eqref{eq:totalC} to $\p = 62.5 > \p^\star$. An encoder-decoder pair under this constraint is designed below.
Specifically, construct a matrix
\begin{equation*}
    Y=\begin{bmatrix}
     3.3201 &   0.7449  &  0.6413  &  1.0293\\
         0 &   2.7183  &  1.2888  &  2.0683\\
         0    &     0  &  1.7333 &  1.8557\\
         0  &       0    &     0  &  1.6487
    \end{bmatrix}
\end{equation*}
that is similar to matrix $A$, with singular values given by $e^{\ca_1^\star}$, $e^{\ca_2^\star}$, $e^{\ca_3^\star}$, and $1$. 
The matrices $\tilde{M}$ and $\tilde{L}$ are then computed from
\eqref{eq:thmMu} and \eqref{eq:thmLu}, respectively.
With $U=I_{3\times 3}$ and $\epsilon = 0.01$,  \eqref{eq:thm_encoderdecoder} gives
 \begin{equation*}
 \begin{aligned}
          \bar{M}&\!=\!\begin{bmatrix}
        -1.8851 &  0 &   2.0809  &  0.8599\\
   -4.3018  &  3.9445 &  -2.4963  &-11.2931\\
         0   & 2.2172 &  -8.3036 &  -6.7746
     \end{bmatrix}\!\times\! 10^{-2}\text{ and}\\
     \bar{L}&\!=\!\begin{bmatrix}
          3.0606 &   4.3741 &   0.7701  &  0.2805\\
   -0.5626  & -1.7848 &  -0.3143 &  -0.1145\\
    0.7865  &  1.6386 &   0.4394  &  0.1600
     \end{bmatrix}^\prime\!\times\! 10^2.
 \end{aligned}
 \end{equation*}
 The encoder and decoder are then implemented as in
Fig.~\ref{fig:encoder-decoder-design}.
Under this design, the poles of the estimation error system \eqref{eq:err_system} are $e^{-1.2}$, $e^{-1}$, $e^{-0.55}$, $e^{-0.5}$ respectively, thereby confirming the boundedness of the EEC. We perform 50,000 Monte Carlo simulations over $T \!=\! 100$ steps, yielding a simulated average total transmission power
\begin{equation*}
    \frac{1}{T}\sum\nolimits_{k=0}^{T-1}\sum\nolimits_{i=1}^{3}\mathbb{E}\{s_i(k)^2\}= 62.4170<\p,
\end{equation*}
which satisfies the total power constraint.

\section{Conclusion}\label{sec:conclusion}
In this paper, we investigated LTI coding and transmitting of a discrete-time LTI vector source over power-constrained parallel Gaussian channels with feedback. Under the individual subchannel power constraints, we established a necessary and sufficient solvability condition in terms of two coupled majorization inequalities, which characterize how the subchannel qualities should be matched with the antistable poles of the source. Under the total channel power constraint, we obtained the minimum total power required for achieving bounded EEC by exploiting partial-order progamming under majorization order. For equal subchannel noise variances, an analytic solution to the minimum required power and optimal power allocation was given. The optimal allocation admits a nice water-filling interpretation; for unequal noise variances, a sequential water-filling algorithm was developed. The corresponding coding design procedures were also provided. 
The current work focuses on cyclic unstable sources. Generalization to general noncyclic sources appears nontrivial and is under current investigation.

\appendix
\subsection{Computation of Meet and Join in $(\mathcal{D}^{\mathit{n}},\preccurlyeq_{\mathrm{w}})$}\label{appendix:join}

For a vector $z\in\mathcal{D}^n$, its cumulative sum curve is defined as the linear interpolation of points $\{(i, \sum_{j=1}^i z_j)\}_{i=0}^n$, denoted by $C_z$.
Since the entries of $z$ are in nonincreasing order, the slopes of the successive line segments of $C_z$ are nonincreasing, and thus $C_z$ is concave.
For two vectors $x,y\in\mathcal{D}^n$, their meet $x\land y$ and join $x\lor y$ can be understood intuitively in terms of their cumulative sum curves, as illustrated in Fig.~\ref{fig:meet_join_example}. 
Specifically, $C_{x\land y}$ is the greatest cumulative sum curve lying below $C_x$ and $C_y$.
Since the pointwise minimum of concave functions remains concave, 
$C_{x\land y}$ can be obtained directly by linearly interpolating the points $\{(i, \min\{\sum_{j=1}^i\! x_j, \sum_{j=1}^i\! y_j\})\}_{i=0}^n$. 
Conversely, $C_{x\lor y}$ is the least cumulative sum curve lying above $C_x$ and $C_y$.
However, the pointwise maximum of concave functions is not necessarily concave.
Therefore, $C_{x\lor y}$ cannot, in general, be obtained simply by interpolating the points $\{(i, \max\{\sum_{j=1}^i\! x_j, \sum_{j=1}^i \!y_j\})\}_{i=0}^n$. 
Instead, it is given by the concave envelope of the resulting interpolated curve.
The corresponding computational procedure is detailed in Algorithm~\ref{alg:join-computation}, adapted from \cite{cicalese2002supermodularity}.

\begin{figure}[htbp]
	\centering
	\includegraphics[width=1\linewidth]{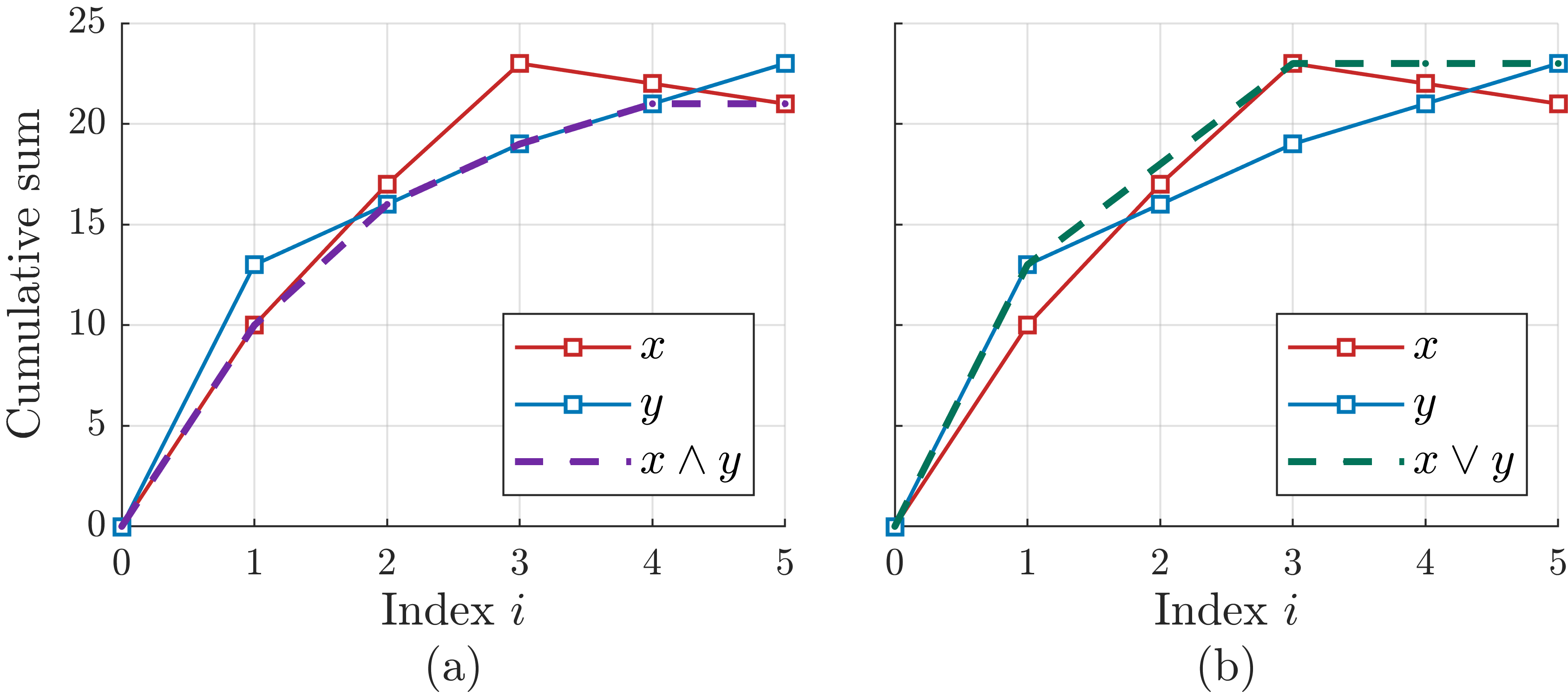}
	\caption{Computation of $x\land y$ and $x\lor y$ for $x=\begin{bmatrix}
	    10\!&\!7\!&\!6\!&\!-1\!&\!-1
	\end{bmatrix}^\prime$ and $y=\begin{bmatrix}
	    13\!&\!3\!&\!3\!&\!2\!&\!2
	\end{bmatrix}^\prime$, yielding $x\land y=\begin{bmatrix}
	    10\!&\!6\!&\!3\!&\!2\!&\!0
	\end{bmatrix}^\prime$ and $x\lor y=\begin{bmatrix}
	    13\!&\!5\!&\!5\!&\!0\!&\!0
	\end{bmatrix}^\prime$.}
    \label{fig:meet_join_example}
\end{figure}

\begin{algorithm}[htbp]\label{algorithm:join}
\SetCommentSty{itshape} 
    \caption{Computation of the join  in $(\mathcal{D}^{n},\preccurlyeq_{\mathrm{w}})$}\label{alg:join-computation}
    \KwIn{\!$x\!=\!\begin{bmatrix}
             x_1\!\!&\!\!x_2\!\!&\!\!\cdots\!\!&\!\!x_n
         \end{bmatrix}^\prime\!\in\!\mathcal{D}^n$, $y\!=\!\begin{bmatrix}
            y_1\!\!&\!\!y_2\!\!&\!\!\cdots\!\!&\!\!y_n
        \end{bmatrix}^\prime\!\in\!\mathcal{D}^n$\!\!\!\!\!\!}
    \KwOut{${x\lor y}\in\mathcal{D}^n$}
     $s_0\gets 0$\;
     \For{$i \gets 1$ \KwTo $n$}{
     $s_i\gets \max\{\sum_{j=1}^ix_j,\sum_{j=1}^iy_j\}$\;
     }
    \tcp{\small Compute the concave envelope of $\{(i,s_i)\}_{i=0}^n$\!\!\!\!\!\!}
     $k\gets n$\;
    \While{$k \geq 1$}{
        $min\_slope \gets \infty$\;
        \For{$i \gets 0$ \KwTo $k-1$}{
            $slope \gets \frac{s_k - s_i}{k - i}$\;
            \If{$slope < min\_slope$}{
                $min\_slope \gets slope$\;
                $j \gets i$\;
            }
        }
        \For{$i \gets j+1$ \KwTo $k$}{
            $z_i \gets min\_slope$\;
        }
        $k \gets j$\;
    }
    \Return{$\begin{bmatrix}
        z_1\!&\!z_2\!&\!\cdots\!&\!z_n
    \end{bmatrix}^\prime$}\;
\end{algorithm}

\subsection{Necessary Conditions for Problem~1}\label{appendix:necessary_con}

\begin{lma}\label{lma:necessary_con}
Consider an encoder-decoder pair that achieves bounded EEC with finite transmission power. Suppose that the encoder realization \eqref{eq:encoder} is stabilizable and detectable.
Then $\begin{bmatrix}
	    A_e & L\\
        M &R
	\end{bmatrix}$ is Schur stable. 
    Moreover, every unstable eigenvalue of $A$ is also an eigenvalue of $A_e$, with algebraic multiplicity no smaller than its algebraic multiplicity in $A$.
\end{lma}

\begin{proof}
    Consider the cascade of the source and encoder, with a loop closed around the encoder via the parallel AWGN channel and delayed channel feedback.
The dynamics of the overall state $z(k):=\begin{bmatrix}
    x(k)^\prime\!&\!
    x_e(k)^\prime\!&\!
    r(k-1)^\prime
\end{bmatrix}^\prime$ are given by
\begin{equation}\label{eq:overall_system}
    \begin{aligned}
        z(k+1)&=A_{\mathrm{cl}} z(k)+B_{\mathrm{cl}} \begin{bmatrix}
            v(k)\\
            w(k)
        \end{bmatrix},\\
        s(k)&=C_{\mathrm{cl}}z(k),
    \end{aligned}
\end{equation}
where
\begin{equation*}
    A_{\mathrm{cl}}=\begin{bmatrix}
        A&0&0\\
        K&A_e&L\\
        N&M&R
    \end{bmatrix},B_{\mathrm{cl}}=\begin{bmatrix}
        I&0\\
        0&0\\
        0&I
    \end{bmatrix},C_{\mathrm{cl}}=\begin{bmatrix}
        N&M&R
    \end{bmatrix}.
\end{equation*}
Since $\left[\!\!\begin{array}{c|cc}A_{e} \!&K\!&\!L\end{array}\!\!\right]$ is stabilizable, $\left[\!\!\begin{array}{c|c}A_{\mathrm{cl}}\!&\!B_{\mathrm{cl}}\end{array}\!\!\right]$ is stabilizable. To ensure finite transmission power ($\lim_{k\rightarrow\infty}\mathbb{E}\{\|s(k)\|^2\}<\infty$), 
all observable modes of \eqref{eq:overall_system} must be stable, and thus
\begin{equation}\label{eq:modes_requirement}
    \lim_{k\rightarrow\infty}C_{\mathrm{cl}}A_{\mathrm{cl}}^k=0.
\end{equation}
Denote 
\begin{equation*}
    A_{\mathrm{fb}}=\begin{bmatrix}
        A_e&L\\
        M&R
    \end{bmatrix}\text{ and }
    C_{\mathrm{fb}}=\begin{bmatrix}
        M &R
    \end{bmatrix}.
\end{equation*}
Since $\left[\!\begin{array}{c}A_{e}\\ \hline M\end{array}\!\right]$ is detectable, $\left[\!\begin{array}{c}A_{\mathrm{fb}}\\ \hline C_{\mathrm{fb}}\end{array}\!\right]$ is detectable. 
To prove that $A_{\mathrm{fb}}$ is Schur stable, assume for contradiction that $A_{\mathrm{fb}}$ has an unstable eigenvalue $\lambda$ with a corresponding eigenvector $\nu$. Then, we have
\begin{equation*}
    \begin{aligned}
        A_{\mathrm{cl}}\begin{bmatrix}
            0\\
            \nu
        \end{bmatrix}=\begin{bmatrix}
            0\\
            A_{\mathrm{fb}}\nu
        \end{bmatrix}=\lambda \begin{bmatrix}
            0\\
            \nu
        \end{bmatrix}\text{ and } 
        C_{\mathrm{cl}}\begin{bmatrix}
            0\\
            \nu
        \end{bmatrix}=C_{\mathrm{fb}}\nu\neq 0,
    \end{aligned}
\end{equation*}
which contradicts \eqref{eq:modes_requirement}.
Thus $A_{\mathrm{fb}}$ is Schur stable.
Next, let $J_u$ denote the direct sum of all Jordan blocks associated with the unstable eigenvalues of $A$, and $X$ be a full-column-rank matrix satisfying $AX=XJ_u$.
 Since the spectra of $A_{\mathrm{fb}}$ and $J_u$ are disjoint, the following Sylvester equation
 \begin{equation*}
     A_{\mathrm{fb}}Y-YJ_u=-\begin{bmatrix}
         K\\
         N
     \end{bmatrix}X
 \end{equation*}
has a unique solution $Y$. This yields
\begin{equation}\label{eq:Jordan_E}
    A_{\mathrm{cl}}\begin{bmatrix}
        X\\
        Y
    \end{bmatrix}=\begin{bmatrix}
        X\\
        Y
    \end{bmatrix}J_u.
\end{equation}
Then by \eqref{eq:modes_requirement}, it is necessary that
\begin{equation*}\label{eq:unobservable_E}
    C_{\mathrm{cl}}\begin{bmatrix}
        X\\
        Y
    \end{bmatrix}=0.
\end{equation*}
 Partition $Y=\begin{bmatrix}
        Y_1\\
        Y_2
    \end{bmatrix}$ compatibly with $A_{\mathrm{fb}}$. 
    From \eqref{eq:Jordan_E}, we then obtain $Y_2J_u=0$, which implies $Y_2=0$. Consequently,
    \begin{equation}\label{eq:unobservable_E2}
			\begin{bmatrix}
				A & 0\\
				K& A_e
			\end{bmatrix}\begin{bmatrix}
            X\\
            Y_1
        \end{bmatrix}=\begin{bmatrix}
            X\\
            Y_1
        \end{bmatrix}J_u\text{ and } \begin{bmatrix}
				N&M
			\end{bmatrix}\begin{bmatrix}
            X\\
            Y_1
        \end{bmatrix}=0.
    \end{equation}
    
We next bridge this with the requirement of bounded EEC.
By superposition, decompose $x_e(k)=\xi_1(k)+\xi_2(k)$,
where 
\begin{align}
    \xi_1(k+1)&=A_e\xi_1(k)+Kx(k), &\xi_1(0)=0,\label{eq:xi_1}\\
    \xi_2(k+1)&=A_e\xi_2(k)+Lr(k-1),&\xi_2(0)=0.\notag\label{eq:xi_2}
\end{align}
Define 
\begin{equation}\label{eq:define_y}
    y(k)=r(k)-R\,r(k-1)-M\xi_2(k).
\end{equation}
Given $r(-1)=0$ and $\xi_2(0)=0$, the signal $r$ can be reconstructed
causally from $y$ through
\begin{equation*}\label{eq:inverse_filter}
    \begin{bmatrix}
        \xi_2(k+1)\\
        r(k)
    \end{bmatrix}
    =
    A_{\mathrm{fb}}
    \begin{bmatrix}
        \xi_2(k)\\
        r(k-1)
    \end{bmatrix}
    +
    \begin{bmatrix}
        0\\I
    \end{bmatrix}y(k).
\end{equation*}
Thus, an LTI estimator achieving bounded EEC for $x$ can be constructed based on $y$.
From \eqref{eq:define_y} and \eqref{eq:overall_system}, we get
    \begin{equation*}
    \begin{aligned}
        y(k)&=s(k)+w(k)-R\,r(k-1)-M\xi_2(k)\\
        &={N}x(k)+Mx_e(k)+w(k)-M\xi_2(k)\\
        &={N}x(k)+M\xi_1(k)+w(k).
    \end{aligned}
    \end{equation*}
    For \eqref{eq:xi_1}, using the controllability decomposition,
    there exists a nonsingular matrix $P$ such that     \begin{equation}\label{eq:controllable-uncontrollable_E}
    P^{-1}A_eP
    =
    \begin{bmatrix}
        A_{11} & A_{12}\\
        0      & A_{22}
    \end{bmatrix},
    \qquad
    P^{-1}K
    =
    \begin{bmatrix}
        K_1\\0
    \end{bmatrix},
\end{equation}
    where $\left[\!\!\begin{array}{c|c}A_{11} &K_1\end{array}\!\!\right]$ is controllable. Partition $P^{-1}\xi_1\!=\!\begin{bmatrix}
        \xi_{11}(k)\\ \xi_{12}(k)
    \end{bmatrix}$ and $MP\!=\!\begin{bmatrix}
        M_1\!&\!M_2
    \end{bmatrix}$ accordingly.
   Since $\xi_{1}(0)=0$, the uncontrollable part satisfies $\xi_{12}(k) \equiv 0$, yielding the reduced system
\begin{equation*}\label{eq:reduced_estimation_system}
\begin{aligned}
    \begin{bmatrix}
        x(k+1)\\
        \xi_{11}(k+1)
    \end{bmatrix}
    &=
    \underbrace{
    \begin{bmatrix}
        A   & 0\\
        K_1 & A_{11}
    \end{bmatrix}}_{\displaystyle A_{\mathrm{rd}}}
    \begin{bmatrix}
        x(k)\\
        \xi_{11}(k)
    \end{bmatrix}
    +
    \underbrace{
    \begin{bmatrix}
        I\\0
    \end{bmatrix}}_{\displaystyle  B_{\mathrm{rd}}}
    v(k),\\
    y(k)
    &=
    \underbrace{
    \begin{bmatrix}
        N & M_1
    \end{bmatrix}}_{\displaystyle C_{\mathrm{rd}}}
    \begin{bmatrix}
        x(k)\\
        \xi_{11}(k)
    \end{bmatrix}
    +w(k).
\end{aligned}
\end{equation*}
    Since $\left[\!\!\begin{array}{c|c}A_{11} \!&\!K_1\end{array}\!\!\right]$ is controllable, $\left[\!\!\begin{array}{c|c}A_{\mathrm{rd}}\!&\!B_{\mathrm{rd}}\end{array}\!\!\right]$ is controllable.
For this system to admit an estimator for $x$ with bounded EEC, by PBH test, it is necessary that for any $\lambda$ with $|\lambda|\geq 1$ and vector $\begin{bmatrix}
        \alpha_1^\prime&\alpha_2^\prime
    \end{bmatrix}^\prime$ with $\alpha_1\neq0$ satisfying 
\begin{equation}\label{eq:necessary:con}
A_{\mathrm{rd}}\!\begin{bmatrix}
             \alpha_1\\
             \alpha_2
         \end{bmatrix}\!=\!\lambda \begin{bmatrix}
             \alpha_1\\
             \alpha_2
         \end{bmatrix},\quad
       C_{\mathrm{rd}}\begin{bmatrix}
             \alpha_1\\
             \alpha_2
         \end{bmatrix}\neq0.
    \end{equation}
    Now, according to  \eqref{eq:controllable-uncontrollable_E}, we partition $P^{-1}Y_1\!=\!\begin{bmatrix}
        Y_{11}\\
        Y_{12}
    \end{bmatrix}$. Applying the similarity transformation induced by $P$ to \eqref{eq:unobservable_E2} yields
         \begin{align}
        A_{\mathrm{rd}}\begin{bmatrix}
            X\\
            Y_{11}
        \end{bmatrix}+\begin{bmatrix}
            0\\
            A_{12}
        \end{bmatrix}Y_{12}&=\begin{bmatrix}
            X\\
            Y_{11}
        \end{bmatrix}J_u,\label{eq:necessary_con_1}\\
        A_{22}Y_{12}&=Y_{12}J_u,\label{eq:necessary_con_3}\\
        C_{\mathrm{rd}}\begin{bmatrix}
            X\\
            Y_{11}
        \end{bmatrix}+M_2Y_{12}&=0.\label{eq:necessary_con_2}
    \end{align}
        From \eqref{eq:necessary_con_3}, it suffices to show that $Y_{12}$ has full column rank. 
        Assume for contradiction that its null space $\mathcal{N}\{Y_{12}\}\neq\{0\}$. 
        By \eqref{eq:necessary_con_3}, for any $\mu\!\in\!\mathcal{N}\{Y_{12}\}$, we have 
        \[Y_{12}J_u\mu=A_{22}Y_{12}\mu
    =0,\]
        implying that $\mathcal{N}\{Y_{12}\}$ is a $J_u$-invariant subspace. Thus, there exists a nonzero vector $\eta\in \mathcal{N}\{Y_{12}\}$ such that $J_u \eta=\lambda\eta$ for some eigenvalue $\lambda$ of $J_u$. Right-multiplying \eqref{eq:necessary_con_1} and \eqref{eq:necessary_con_2} by $\eta$, and noting $Y_{12}\eta = 0$, we obtain
\begin{equation*}\label{eq:unobservable}
        A_{\mathrm{rd}}\begin{bmatrix}
        X\eta\\
        Y_{11}\eta
    \end{bmatrix}=\lambda \begin{bmatrix}
        X\eta\\
        Y_{11}\eta
    \end{bmatrix}\text{ and }
        C_{\mathrm{rd}}\begin{bmatrix}
        X\eta\\
        Y_{11}\eta
    \end{bmatrix}=0,
    \end{equation*}
    Since $X$ has full column rank, $X\eta\neq0$, thereby contradicting \eqref{eq:necessary:con}. This completes the proof.
\end{proof}

\subsection{Technical Lemmas}\label{appdix:key_lemma}

Let $\left[\begin{array}{c}
    \hua{A}  \\ \hline
     \hua{C} 
\end{array}\right]$ be a detectable discrete-time LTI system, with $\hua{A}\in\mathbb{R}^{n\times n}$ and $\hua{C}\in\mathbb{R}^{q\times n}$. The associated DARE 
\begin{equation}\label{eq:dare}
	X=\hua{A}X(I+\hua{C}'\hua{C}X)^{-1}\hua{A}'
\end{equation}
admits a unique real symmetric semi-stabilizing solution $X$ such that all eigenvalues of $\hua{A}-\hua{A}X\hua{C}'(I+\hua{C}X\hua{C}')^{-1}\hua{C}$ lie in the closed unit disk \cite{de2003riccati}.

The following lemma establishes a log-majorization relation connecting the antistable eigenvalues $\lambda_1,\lambda_2,\dots,\lambda_{m}$ of $\hua{A}$ with the eigenvalues $\mu_1\geq\mu_2\geq\cdots\geq\mu_q$ of $I+\hua{C}X\hua{C}'$.

\begin{lma}\label{lma:major_inequality}
     There holds        \begin{align}
			\!\!\!&\begin{bmatrix}
				|\lambda_1|^2&|\lambda_2|^2&\cdots&|\lambda_{m}|^2
			\end{bmatrix}' \underset{\mathrm{log}}\preccurlyeq \notag\\
            \!\!\!&\qquad\qquad\quad\begin{bmatrix}
				\mu_1\!&\!\mu_2\!&\!\cdots\!&\!\mu_{\min\{q,m\}} \!&\!1 \!&\!\cdots\!&\!1
			\end{bmatrix}'.\label{eq:lma_majorization}
	\end{align}
\end{lma}
\begin{proof}
	Without loss of generality, assume that $\left[\begin{array}{c}
    \hua{A}  \\ \hline
     \hua{C} 
\end{array}\right]$ is in the following decomposition form
	\begin{equation*}
        \left[\begin{array}{c}
    \hua{A}  \\ \hline
     \hua{C} 
\end{array}\right]=\left[\begin{array}{cc}
    \hua{A}_s &0  \\
    0 & \hua{A}_u\\ \hline
     \hua{C}_s& \hua{C}_u 
\end{array}\right],
    \end{equation*}
	where $\hua{A}_u\in\mathbb{R}^{m\times m}$ contains precisely the antistable eigenvalues of $\hua{A}$, and $\hua{A}_s$ contains the remaining eigenvalues.
    Since all eigenvalues of $\hua{A}_u$ lie outside the unit circle and $\left[\begin{array}{c}
    \hua{A}_u  \\ \hline
     \hua{C}_u 
\end{array}\right]$ is observable, 
the DARE
	\begin{equation}\label{eq:lma_dare}
X_u=\hua{A}_uX_u(I+\hua{C}_u'\hua{C}_uX_u)^{-1}\hua{A}_u'
	\end{equation} 
    admits a positive definite stabilizing solution $X_u$ \cite{zhou1996robust}. 
    It can be verified that the semi-stabilizing solution of \eqref{eq:dare} is precisely given by 
    $X=\begin{bmatrix}
			0&0\\
			0&X_u
		\end{bmatrix}$.
   	Let $X_u=ZZ'$ with $Z\in\mathbb{R}^{m\times m}$. 
    Then 
$\hua{C}X\hua{C}'=\hua{C}_uX_u\hua{C}_u'=\hua{C}_uZZ'\hua{C}_u'$. 
Since $\hua{C}_uZZ'\hua{C}_u'$ and $Z'\hua{C}_u'\hua{C}_uZ$ share the same nonzero eigenvalues,  
    the eigenvalues of $I+Z'\hua{C}_u'\hua{C}_uZ$ are $\mu_1,\dots,\mu_{\min\{q,{m}\}}$ together with ${m}-\min\{q,{m}\}$ ones.
    From \eqref{eq:lma_dare}, we have
	\begin{gather}
		ZZ'=\hua{A}_uZZ'(I+\hua{C}_u'\hua{C}_uZZ')^{-1}\hua{A}_u',\notag\\
		I=Z^{-1}\hua{A}_uZ(I+Z'\hua{C}_u'\hua{C}_uZ)^{-1}Z'\hua{A}_u'(Z^{-1})',\notag\\
		Z'\hua{A}_u'(Z^{-1})'Z^{-1}\hua{A}_uZ=I+Z'\hua{C}_u'\hua{C}_uZ, \notag\label{eq:lma_riccati}
	\end{gather}
    which implies that the eigenvalues of $I+Z'\hua{C}_u'\hua{C}_uZ$ are squares of the singular values of $Z^{-1}\hua{A}_uZ$.
    Note that the eigenvalues of $Z^{-1}\hua{A}_uZ$ are $\lambda_1,\lambda_2,\dots,\lambda_{{m}}$. Applying Lemma~\ref{lma:cyclic_matrix} yields \eqref{eq:lma_majorization}, which completes the proof.
\end{proof}

The following lemma is useful in proving Theorem~\ref{thm:least}.

\begin{lma}\label{lma:w_f}
    Let $x,y\in\mathbb{R}^n$ and  $\eta>0$. If $x\preccurlyeq_{\mathrm{w}}y$ and $\sum_{i=1}^{n}x_i+\eta\leq \sum_{i=1}^{n}y_i$, then 
    \begin{equation*}\label{eq:lma:max}
        \begin{aligned}
           \begin{bmatrix}
		\max\{x_1,h\}&\max\{x_2,h\}&\cdots&\max\{x_n,h\}
	\end{bmatrix}'\preccurlyeq_{\mathrm{w}}y,
        \end{aligned}
    \end{equation*}
    where $h$ is the solution to $\sum_{i=1}^{n}\max\{0,h-x_i\}=\eta$.
\end{lma}

\begin{proof}
    Without loss of generality, assume $x,y\in\mathcal{D}^n$.
    Denote $\bar{x}=\begin{bmatrix}
		\max\{x_1,h\}\!&\!\cdots\!&\!\max\{x_n,h\}
	\end{bmatrix}'$.
     Then $\bar{x}\in\mathcal{D}^n$, 
     implying that there exists an index $t\in\{1,\dots,n\}$ such that
	\begin{equation*}
		\begin{aligned}
			\bar{x}_i=\begin{cases}
				x_i,&1\leq i<t, \\
				h, &t\leq i\leq n.
			\end{cases}
		\end{aligned}
	\end{equation*}
 	Given $x\preccurlyeq_{\mathrm{w}}y$, we obtain that $\sum_{i=1}^{k}\bar{x}_i=\sum_{i=1}^{k}x_i\leq\sum_{i=1}^{k}y_i$ for $1 \leq k < t$.  
    Moreover,
    \begin{equation*}
    \sum_{i=1}^{n}\bar{x}_i=\sum_{i=1}^{n}(x_i+\max\{0,h-x_i\})=\sum_{i=1}^{n}x_i+\eta\leq \sum_{i=1}^{n}y_i.
    \end{equation*}
    It remains to show $\sum_{i=1}^{k}\bar{x}_i\leq\sum_{i=1}^{k}y_i$ for $t\leq k\leq n-1$.
	Suppose, for contradiction, that this does not hold.
    Let $m \in \{t, \dots, n-1\}$ be the smallest index such that $\sum_{i=1}^{m} \bar{x}_i > \sum_{i=1}^{m} y_i$. Since $\sum_{i=1}^{m-1}\bar{x}_i\leq\sum_{i=1}^{m-1}y_i$, we have $\bar{x}_m > y_m$.
	Together with $\bar{x}_m=\cdots=\bar{x}_n=h$ and $y\in\mathcal{D}^n$, we obtain 
    \begin{equation*}
        \sum_{i=1}^{n}\bar{x}_i=\sum_{i=1}^{m}\bar{x}_i+(n-m)\bar{x}_m>\sum_{i=1}^{m}y_i+(n-m)y_m\geq\sum_{i=1}^{n}y_i,
    \end{equation*}
    which contradicts $\sum_{i=1}^{n}\bar{x}_i\leq\sum_{i=1}^{n}y_i$.
    This completes the proof.
\end{proof}

\section*{References}
\bibliographystyle{IEEEtran}
\bibliography{my_bib}

\end{document}